\UseRawInputEncoding
\documentclass[lettersize,journal]{IEEEtran}

  \usepackage[pdftex]{graphicx}
  \DeclareGraphicsExtensions{.pdf,.jpeg,.png}
\usepackage{algorithm}
\usepackage{algpseudocode}

\usepackage{etoolbox}
\usepackage{tabu}
\usepackage{setspace}
\usepackage{makecell}
\usepackage{epsfig}
\usepackage[usenames, dvipsnames]{color}
\usepackage[roman]{parnotes}
\usepackage{flafter}
\usepackage[section]{placeins}

\usepackage{amsthm,amsmath,amssymb}
\usepackage{amsmath}
\usepackage{mathtools}
\usepackage{breqn}
\usepackage{amsfonts}
\usepackage{amssymb}
\usepackage{amsthm}
\usepackage{mathrsfs}
\usepackage{xparse}
\usepackage{multirow}
\usepackage{chngcntr}
\usepackage{dblfloatfix}    
\usepackage{soul} 
\usepackage{cite}
\usepackage{bm}
\usepackage{footnote}
\usepackage{subcaption}
\usepackage{algorithmicx}
\usepackage{algpseudocode}
\usepackage{soul}

\algdef{SE}[DOWHILE]{Do}{doWhile}{\algorithmicdo}[1]{\algorithmicwhile\ #1}%

\makeatletter
\newenvironment{breakablealgorithm}
  {
   \begin{center}
     \refstepcounter{algorithm}
     \hrule height.8pt depth0pt \kern2pt
     \renewcommand{\caption}[2][\relax]{
       {\raggedright\textbf{\fname@algorithm~\thealgorithm} ##2\par}%
       \ifx\relax##1\relax 
         \addcontentsline{loa}{algorithm}{\protect\numberline{\thealgorithm}##2}%
       \else 
         \addcontentsline{loa}{algorithm}{\protect\numberline{\thealgorithm}##1}%
       \fi
       \kern2pt\hrule\kern2pt
     }
  }{
     \kern2pt\hrule\relax
   \end{center}
  }
\makeatother

\makeatletter
\newcommand*{\algrule}[1][\algorithmicindent]{%
  \makebox[#1][l]{%
    \hspace*{.2em}
    \vrule height .75\baselineskip depth .25\baselineskip
  }
}

\newcount\ALG@printindent@tempcnta
\def\ALG@printindent{%
    \ifnum \theALG@nested>0
    \ifx\ALG@text\ALG@x@notext
    \else
    \unskip
    \ALG@printindent@tempcnta=1
    \loop
    \algrule[\csname ALG@ind@\the\ALG@printindent@tempcnta\endcsname]%
    \advance \ALG@printindent@tempcnta 1
    \ifnum \ALG@printindent@tempcnta<\numexpr\theALG@nested+1\relax
    \repeat
    \fi
    \fi
}
\patchcmd{\ALG@doentity}{\noindent\hskip\ALG@tlm}{\ALG@printindent}{}{\errmessage{failed to patch}}
\patchcmd{\ALG@doentity}{\item[]\nointerlineskip}{}{}{} 
\makeatother

\DeclareMathOperator{\tr}{tr}
\DeclareMathOperator{\diag}{diag}

\newcommand{\argmaxF}{\mathop{\mathrm{argmax}}\limits}
\newcommand{\argminF}{\mathop{\mathrm{argmin}}\limits}
\allowdisplaybreaks

\begin{document}
\title{An information-theoretic branch-and-prune algorithm for discrete phase optimization of RIS in massive MIMO}
%
\makeatletter
\def\ps@IEEEtitlepagestyle{%
  \def\@oddhead{\mycopyrightnotice}%
  \def\@oddfoot{\hbox{}\@IEEEheaderstyle\leftmark\hfil\thepage}\relax
  \def\@evenhead{\@IEEEheaderstyle\thepage\hfil\leftmark\hbox{}}\relax
  \def\@evenfoot{}%
}
\def\mycopyrightnotice{%
  \begin{minipage}{\textwidth}
  \scriptsize
  Copyright~\copyright~2022 IEEE. Personal use of this material is permitted. Permission from IEEE must be obtained for all other uses, in any current or future media, including\\reprinting/republishing this material for advertising or promotional purposes, creating new collective works, for resale or redistribution to servers or lists, or reuse of any copyrighted component of this work in other works by sending a request to pubs-permissions@ieee.org. Accepted for publication in IEEE Transactions on Vehicular Technology
  \end{minipage}
}

\author{I. Zakir Ahmed,~\IEEEmembership{Member,~IEEE,} Hamid R. Sadjadpour,~\IEEEmembership{Senior Member,~IEEE,} and Shahram Yousefi,~\IEEEmembership{Senior Member,~IEEE}
\thanks{I. Zakir Ahmed and Hamid R. Sadjadpour are with the Department of Electrical and Computer Engineering, University of California at Santa Cruz, Santa Cruz, CA 95064 USA. Shahram Yousefi is with the Department of Electrical and Computer Engineering, Queen’s University, Kingston, ON K7L 3N6, Canada.}}
\markboth{}%
{Shell \MakeLowercase{\textit{et al.}}: A Sample Article Using IEEEtran.cls for IEEE Journals}
%

\maketitle
\begin{abstract}
In this paper, we consider passive RIS-assisted multi-user communication between wireless nodes to improve the blocked line-of-sight (LOS) link performance. The wireless nodes are assumed to be equipped with Massive Multiple-Input Multiple-Output antennas, hybrid precoder, combiner, and low-resolution analog-to-digital converters (ADCs). We first derive the expression for the Cramer-Rao lower bound (CRLB) of the Mean Squared Error (MSE) of the received and combined signal at the intended receiver under interference. By appropriate design of the hybrid precoder, combiner, and RIS phase settings, it can be shown that the MSE achieves the CRLB. We further show that minimizing the MSE w.r.t. the phase settings of the RIS is equivalent to maximizing the throughput and energy efficiency of the system. We then propose a novel Information-Directed Branch-and-Prune (IDBP) algorithm to derive the phase settings of the RIS. We, for the first time in the literature, use an information-theoretic measure to decide on the pruning rules in a tree-search algorithm to arrive at the RIS phase-setting solution, which is vastly different compared to the traditional branch-and-bound algorithm that uses bounds of the cost function to define the pruning rules. In addition, we provide the theoretical guarantees of the near-optimality of the RIS phase-setting solution thus obtained using the Asymptotic Equipartition property. This also ensures near-optimal throughput and MSE performance.\\
\end{abstract}



\section{Introduction}
The Reconfigurable Intelligent Surfaces (RIS) are known to mitigate the harsh effects of wireless channels such as obstruction, shadowing, fading, and other complex scenarios encountered between the transmitter and receiver of interest. This is achieved by efficient beamforming and interference management by the RIS. The RIS comprises an array of large number of reflecting elements, each of which can be controlled to change the amplitude, delay (phase shift), and polarization of the incident signal from the transmitter. In the case of passive RIS structures, only the phase of the incident signal is changed, and the RIS consumes no power in such a situation. In one of the typical architectures, the desired phase shift to be induced upon the incident signal can be achieved by controlling the bias voltage to the positive-intrinsic-negative (PIN) diode associated with each of the RIS elements \cite{QingTut}.\\
\indent The vehicular communication frameworks, namely Vehicle to Everything (V2X) based on the IEEE 802.11p Wireless Local Area Network (WLAN) and the Cellular-V2X (C-V2X) defined by the 3GPP and 5G Automotive Association (5GAA) aim to achieve the goals of the Intelligent Transportation Systems (ITS) \cite{IEEE801p,3GPPcv2x}. The objectives of the ITS include collision avoidance, ease road congestion, accident information, pedestrian safety, emergency vehicle approach warning, and parking assistance, to name a few. With the adoption of massive Multiple-Input Multiple-Output (MaMIMO), millimeter-wave (mmWave), and Terahertz (THz) communications in the next generations of wireless communication, it is natural that the vehicular communication nodes will encompass them in the future. A millimeter MaMIMO framework for C-V2X is proposed and studied in \cite{RisVeh1}. Vehicular wireless links are prone to significant challenges due to the highly dynamic nature of the channels due to large buildings, continuous traffic, and changing landscapes. The integration of the RIS technology to vehicular communication is being studied in the literature and has shown promising results. They are shown to maximize the sum V2X link capacity while guaranteeing the minimum SINR of the vehicle-to-vehicle links \cite{VehicleIRS1,VehicleIRS2,VehicleIRS3}.\\
\indent RIS is one of the key enablers for the sixth-generation (6G) mobile communication networks. This is particularly useful for problems of coverage extension in mmWave and THz communication systems due to the unfavorable free-space omnidirectional path loss in these frequency bands \cite{AppIRS,6GIrs}. In addition to enhancing the wireless link's performance between the transmitter and receiver, the RIS has found applications in providing physical layer security. Advanced signal processing techniques are used to manipulate the wireless channel using RIS to guarantee the security of the communication content in an information-theoretic sense. Essentially the RIS ascertains physical-layer security by configuring the RIS elements in such a fashion to add the wireless signals constructively to the legitimate receiver but destructively to a potential eavesdropper \cite{SecIRS}. A few other examples of the applications of RIS include enhancing the link performance of the cell-edge users who suffer high signal attenuation from the base station (BS), co-channel interference from near BSs \cite{edgeIRS,nomaIRS}, Interference management to support low-power transmission to enhance individual data links in device-to-device networks  \cite{d2dIRS}, In non-orthogonal multiple-access (NOMA) systems, RIS could be considered to increase the number of served users and enhance the rate of communication, which constitutes the major requirement to be accomplished in these systems \cite{nomaIRS,nomaIRS_1,starRIS}. Improve the link performance between the unmanned aerial vehicle (UAV) network and the ground users for UAV trajectory optimization and improve overall system performance, including energy efficiency \cite{uavIRS}.\\
\indent The fundamental problem in all of the above applications is configuring the RIS phase-shift setting to achieve a specific goal. Finding an optimal RIS configuration for a set of $K$ discrete phase shifts with $M$ element array has an exponential time complexity $O(K^M)$. In addition, the objective function is often non-convex in the decision parameters (RIS phase shift settings). Identifying the optimal RIS phase shift is a non-convex NP-Hard combinatorial optimization problem.\\
\indent The proposed algorithm benefits several similar problems related to the wireless backhaul link in the vehicular network or the roadside unit layer, vehicle-to-everything framework, and cellular systems, to name a few.

\subsection{Related works}
Previously, a branch-and-bound (BnB) algorithm was used to solve an optimization problem involving RIS phase shifts to maximize the spectral efficiency (SE) \cite{bnbconvex}. A block-coordinated descent algorithm to maximize the achievable uplink rate with multiple single-antenna users and multi-antenna base stations was proposed in \cite{VarADCIRS}. There, resolution-adaptive analog-to-digital converters (ADCs) operating at millimeter-wave (mmWave) frequencies were assisted by a passive RIS. A trace-maximization-based optimization framework was presented in \cite{RisASyed} to study the effect of the link capacity in a point-to-point MIMO link that considers two RIS architectures. A trellis-based joint optimization of the beamformer and the RIS discrete phase shifts to minimize the mean squared error (MSE) of the received symbols was proposed in \cite{trellisIRS}. In \cite{IRSnoma}, a RIS-assisted architecture is proposed to maximize channel power gains between two users in a NOMA framework. A branch-and-bound (BnB) algorithm is used to solve an optimization problem involving RIS phase shifts to maximize the spectral efficiency (SE) in \cite{bnbconvex}. The solution obtained is achieved by linear approximation of the objective function involving the phase shifts of the RIS. Also, the SE maximization is accomplished by relaxing it to a convex problem. RIS-assisted optimal beamforming for a Multiple-Input Single-Output (MISO) communication system is proposed in \cite{manifoldIRS}. An optimal global solution using BnB is claimed in it. However, the results obtained are not compared with the exhaustive search technique. In addition, the bounds for the BnB algorithm are obtained using convex approximations. The authors in \cite{DiscPhIRS1} propose a low-complexity algorithm using alternating optimization (AO) to jointly optimize transmit-beamforming and RIS phase shift settings to minimize the transmit power from the multi-antenna access point (AP) to multiple single-antenna users. A RIS-aided point-to-point multi-data-stream MIMO is studied in \cite{JointRISPrec}. An AO-based algorithm to jointly optimize the RIS phase shifts and precoder is investigated in it to minimize the symbol rate error. However, the combiner design is not considered in this work. Also, the optimality guarantees of the proposed AO algorithm are not investigated. In \cite{IoTIRS}, a RIS-aided MIMO simultaneous wireless information and power transfer (SWIPT) for Internet-of-Things (IoT) networks are investigated. A BnB algorithm is proposed to maximize the minimum signal-to-interference-plus-noise ratio (SINR) among all information decoders (IDs) while maintaining the minimum total harvested energy at all energy receivers (ERs). In it, the authors relax the quadratic assignment problem to a linear integer problem and use the BnB method to obtain the solution. A joint multi-UAV trajectory/communication optimization problem in a network with RISs on uneven terrain is proposed in \cite{MultUAV}. An effective path-planning algorithm for this optimization problem is proposed. Although the paper deems that the issue of RIS control (either phase-shift or amplitude) is beyond its scope, a mathematically rigorous proof of its asymptotic optimality is given. However, the problem considered is a continuous optimization problem, and the computational complexity of the approach is not discussed in it. An asymptotic analysis for RIS assisted communication between multi-antenna users for mmWave MaMIMO is studied in \cite{intfDian1}. The problem of minimizing the transmit power subject to the rate constraint is also analyzed for the scenario without direct paths in the pure LOS propagation.\\
\indent All the earlier works in literature make convex approximations of the objective function under consideration and solve the same using various well-established algorithms, for example, Branch-and-Bound (BnB). However, to the best of our knowledge, none of the earlier works show theoretical guarantees for either optimality or near-optimality, considering the original non-convex problem. 

\subsection{Our contribution}
In this paper, we consider a RIS-assisted multi-user MaMIMO communication system under interference. The contributions of this paper are as follows:
\textit{(i)} We first derive the expression for the Cramer-Rao lower bound (CRLB) of the MSE of the received and combined signal as a function of the phase settings of the RIS for a given hybrid precoder, combiner, and ADC bits,\\
\textit{(ii)} we show that minimizing the MSE by adjusting the RIS phase shifts also ensures maximization of throughput and energy efficiency,\\
\textit{(iii)} we show that the MSE achieves the CRLB with the appropriate design of the hybrid precoders, and combiners,\\
\textit{(iv)} we present a novel Information-Directed Branch-and-Prune (IDBP) algorithm, in which, we, to the best of our knowledge, for the first time in the literature use an information-theoretic measure to decide on the pruning rules in a tree-search algorithm to arrive at the RIS phase-setting solution, which is vastly different compared to the traditional branch-and-bound algorithm that uses bounds of the cost function to define the pruning rules.\\
\textit{(v)} we establish theoretical guarantees for near-optimality, and substantiate the claims by comparing the solutions obtained with the exhaustive search method for a smaller value of $M$. 
\textit{(vi)} We compare the performance and the time complexity of the proposed algorithm with the state-of-the-art trace-maximization-based approach for MIMO transceiver structure proposed in \cite{RisASyed}, and the AO algorithm proposed in \cite{JointRISPrec}, both for larger $M$.

\subsection{Notations}
The column vectors are represented as boldface small letters and matrices as boldface uppercase letters. The primary diagonal of a matrix is denoted as ${\rm diag}(\cdot)$, and all expectations $E[\cdot]$ are over the random variable $\bold{n}$, which is an AWGN vector. The multivariate normal distribution with mean $\boldsymbol{\mu}$ and covariance $\boldsymbol{\varphi}$ is denoted as $\mathcal{N}(\boldsymbol{\mu},\boldsymbol{\varphi})$ and $\mathcal{CN}(\bold{0},{\boldsymbol{\varphi}})$ denotes a multivariate circularly-symmetric Gaussian distribution. The trace of a matrix $\bold{A}$ is shown as $\tr{(\bold{A})}$ and $\bold{I}_N$ represents a $N \times N$ identity matrix. The frobenius norm of matrix $\bold{A}$ is indicated as $\lVert \bold{A} \rVert_F$. The superscripts T and H denote transpose and Hermitian transpose, respectively. The terms $\mathbb{R}$, and $\mathbb{C}$ indicate the space of real, and complex numbers, respectively.

\subsection{Paper content}
The rest of this paper is organized as follows. Section \ref{sigmod} describes the system model and parameters. In Section \ref{PrecCombDesgn}, we describe the hybrid precoder and combiner design. We discuss the RIS phase shift optimization and derive the optimization framework in Section \ref{ris_opt}. Section \ref{ris_opt} also details the design to fine-tune the digital precoders and combiners. We describe the theoretical framework of the proposed IDBP algorithm in Section \ref{idbp_algo}, including the optimality analysis. The proposed IDBP algorithm is detailed in Section \ref{AlgoDes}. The computational complexity analysis is described in Section \ref{CCA}, followed by simulation results in Sections \ref{Sim}, and conclusions in Section \ref{Conc} respectively. Supporting Theorems and their proofs are presented in the Appendices.

\section{Signal Model}\label{sigmod}
We consider a RIS equipped with $M$ passive reflecting elements each of which can be set to $K$ discrete phase-shift values to aid the millimeter-wave (mmWave) Massive Multiple-Input Multiple-Output (MaMIMO) communication between two roadside units (RSU) in a vehicular wireless backhaul network, typically called the RSU-to-RSU wireless link. In addition, we consider the RSUs to be equipped with hybrid precoders, hybrid combiners, and low-resolution ADCs. The communication is assumed to have a blocked line-of-sight (LoS) signal to the intended RSU receiver. An example use-case scenarios is illustrated using Fig.\ref{fig1} \cite{RisVeh1}. This proposed signal model can be extended to other use-case scenario like V2X, cellular wireless backhaul network nodes without loss of generality.\\
\begin{figure}[ht]
\centering
\includegraphics[scale=0.31]{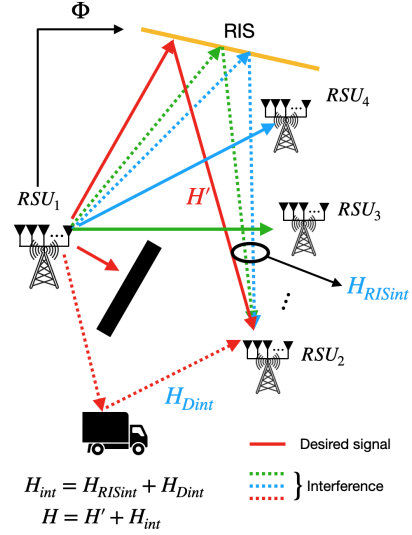}
\caption{An example of an RSU layer employing a RIS for enhancing the performance of a blocked LOS link under Interference.}
\label{fig1}
\end{figure}
\begin{figure}[ht]
\centering
\includegraphics[scale=0.34]{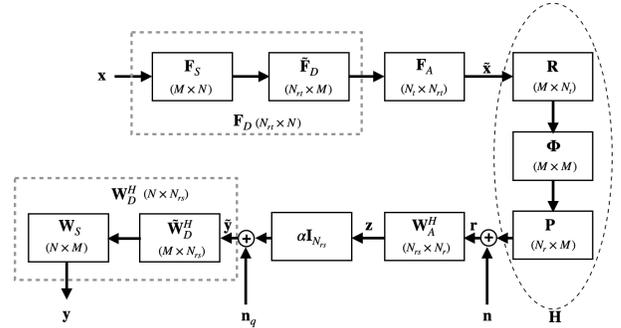}
\caption{System model with RIS-assisted channel with interference.}
\label{fig2}
\end{figure}
\indent The signal model of such a communication system is shown in Fig.\ref{fig2}. Here, we denote ${\bold{F}_D}$ and ${\bold{F}_A}$ to be the digital and analog precoders, respectively. Similarly, we represent ${\bold{W}_D^H}$ and ${\bold{W}_A^H}$ to be the digital and analog combiners, respectively. The vector $\bold{x}$ is an $N \times1$ transmitted signal vector whose average power is unity. Let $N_{rt}$ and $N_{rs}$ denote the number of RF Chains at the transmitter and receiver, respectively. Also, $N_t$ and $N_r$ represent the number of transmit and receive antennas, respectively. The effective channel $\bold{H}$ which is a $N_r \times N_t$ matrix at the intended receiver will be a combination of the RIS reflected signal from the transmitter and the interference of the signal from the same transmitter intended for the other multi-antenna users. That is 
\begin{equation}\label{eq_aa1}
\begin{split}
\bold{H} = \bold{H'} + \bold{H}_{int},
\end{split}
\end{equation}
where the channel $\bold{H'}$ represents the blocked LOS channel between TX and RX assisted by the RIS in the absence of interference, and can be expressed as $\bold{H'} = \bold{Q}\bold{\Phi}\bold{G}$. The term $\bold{G} \in \mathbb{C}^{M \times N_t}$ is the transimtter-to-RIS (TX-RIS) channel, $\bold{Q} \in \mathbb{C}^{N_r \times M}$ being the RIS-to-receiver(RIS-RX) channel \cite{RisASyed}. The action of the $M$ element RIS is represented as $\bold{\Phi} = \diag( e^{\phi_1}, e^{\phi_2}, \cdots, e^{\phi_M})$. Here $\phi_n \in \Phi$, where $\Phi$ is a finite set phase angles with cardinality $K$.\\
\indent The interference channel $\bold{H}_{int}$ represents the combination of the RIS reflected signals from the transmitter to the other users but arriving at the intended receiver $\bold{H}_{RISint}$, and the non-LOS reflected from the transmitter to the receiver not going through the RIS $\bold{H}_{Dint}$ (See Fig. \ref{fig1}). Formally
\begin{equation}\label{eq_aa2}
\begin{split}
\bold{H}_{int} = \bold{H}_{RISint} + \bold{H}_{Dint} = \sum_{i=1}^{\beta} \bold{Q}_i\bold{\Phi}\bold{G}_i + \bold{H}_{Dint},
\end{split}
\end{equation}
where the components $\{ \bold{G}_i \}_{i=1}^{\beta}$ represent the transmitter-to-RIS of the interferers, similarly $\{ \bold{Q}_i \}_{i=1}^{\beta}$ are the RIS-to-receiver channels of the interferers, with ${\beta}$ indicating the total number of interferers. Hence the effective channel between the TX and the RX considering the interferers can be written as
\begin{equation}\label{eq_aa3}
\begin{split}
\bold{H} = \bold{H'} + \bold{H}_{int} = \bold{Q}\bold{\Phi}\bold{G} + \sum_{i=1}^{\beta} \bold{Q}_i\bold{\Phi}\bold{G}_i + \bold{H}_{Dint}.
\end{split}
\end{equation}
Inspired by the channel model adopted in \cite{RisASyed, intfDian1}, we express the RIS-assisted channel with interference given in \eqref{eq_aa3} as a traditional mmWave MIMO channel comprising of $\gamma$ paths (here $\gamma = \beta + 2$, see \eqref{eq_aa3})
\begin{equation}\label{eq_aa4}
\begin{split}
\bold{H} = \bold{A}_r \bold{D}  \bold{A}_t^H,
\end{split}
\end{equation}
where $\bold{D}$ is a $\gamma \times \gamma$ diagonal matrix comprising of the complex gains $\{ \alpha_i \}_{i=1}^{\gamma}$, the matrices $\bold{A}_r$ and $\bold{A}_t$ correspond to the collection of the steering vectors $\bold{a}_r(\phi_r), \bold{a}_t(\theta_t)$ with $\phi_r^i$ and $\theta_t^i$ indicating the angles of arrival and departures respectively. That is
\begin{equation}\label{eq_aa5}
\begin{split}
\bold{A}_r &= [\bold{a}_r(\phi_r^1), \bold{a}_r(\phi_r^2), \cdots, \bold{a}_r(\phi_r^{\gamma})],\\
\bold{A}_t &= [\bold{a}_t(\theta_t^1), \bold{a}_t(\theta_t^2), \cdots, \bold{a}_t(\theta_t^{\gamma})].
\end{split}
\end{equation}
Now, when we choose the number of TX antennas $N_t$ and RX antennas $N_r$ to be very large, the Singular Value Decomposition (SVD) of the matrix $\bold{H}$ in \eqref{eq_aa4} can be shown as \cite{intfDian1,intfDian2,intfAyach}
\begin{equation}\label{eq_aa6}
\begin{split}
\bold{H} = \bold{U}\bold{\Sigma}\bold{V}^H = [\bold{A}_r|\bold{A}_r^{\perp}]\bold{\Sigma}[\bold{\tilde{A}}_t|\bold{\tilde{A}}_t^{\perp}]^H,
\end{split}
\end{equation}
where $\Sigma$ is a diagonal matrix comprising of the singular values on its diagonal
\begin{equation}\label{eq_aa7}
\begin{split}
[\Sigma]_{ii}= 
\begin{cases}
    |\alpha_i|,& \text{ for }1 \le i \le \gamma\\
    0,              & \text{ for }i > \gamma,
\end{cases}
\end{split}
\end{equation}
and the matrix
\begin{equation}\label{eq_aa8}
\begin{split}
\bold{\tilde{A}}_t &= [e^{j\zeta_1}\bold{a}_t(\theta_t^1), e^{j\zeta_2}\bold{a}_t(\theta_t^2), \cdots, e^{j\zeta_{\gamma}}\bold{a}_t(\theta_t^{\gamma})],
\end{split}
\end{equation}
where $\zeta_i$ is the phase component of the complex gain $\alpha_i$. Taking into account the action of the RIS phase shifts $\bold{\Phi}$, we can rewrite \eqref{eq_aa6} as
\begin{equation}\label{eq_aa9}
\begin{split}
\bold{H} &= \bold{U}\bold{\Sigma}\bold{V}^H = [\bold{A}_r|\bold{A}_r^{\perp}]\bold{\Sigma}[\bold{\tilde{A}}_t|\bold{\tilde{A}}_t^{\perp}]^H,\\
&= \bold{U} \bold{\Sigma} \bold{\Phi} \bold{R} = \bold{P} \bold{\Phi} \bold{R},
\end{split}
\end{equation}
where $\bold{R} = \diag(e^{j\zeta_1},e^{j\zeta_2},\cdots,e^{j\zeta_{(\beta + 1)}}\cdots)\bold{V}^H$, and $\bold{P} = \bold{U}\bold{\Sigma}$. It is to be noted that $\bold{P}$ and $\bold{R}$ are not unitary matrices anymore. Hence we can visualize the effective channel $\bold{H}$ as 
\begin{equation}\label{eq_1}
\begin{split}
\bold{H} = \bold{P}\bold{\Phi}\bold{R}.
\end{split}
\end{equation}
$\bold{R} \in \mathbb{C}^{M \times N_t}$ is the effective transimtter-to-RIS (TX-RIS) channel, $\bold{P} \in \mathbb{C}^{N_r \times M}$ the RIS-to-receiver(RIS-RX) channel, both considering the interference.

In this work, we focus on minimizing the mean squared error (MSE) performance of the communication link by optimizing the phase shifts.
The transmitted signal $\bold{\tilde{x}}$ and the received signal $\bold{r}$ are represented as 
\begin{equation}\label{eq1a}
{\bold{\tilde{x}}} = {\bold{F}_A}{\bold{F}_D}{\bold{x}},\text{ }{\bold{r}} = {\bold{H}}{\bold{\tilde{x}}}+{\bold{n}}.
\end{equation}

Here, ${\bold{n}}$ is an $N_r\times1$ noise vector of independent and identically distributed (i.i.d) complex Gaussian random variables such that ${\bold{n}} \sim \mathcal{CN}(\bold{0},{\sigma_n^2}{\bold{I}_{N_r}})$. The received symbol vector $\bold{r}$ is analog-combined with ${\bold{W}_A^H}$ to get ${\bold{z}} = {\bold{W}_A^H}{\bold{r}}$ and later digitized using a low-resolution ADCs to produce $\bold{\tilde{y}} = \bold{Q}_b(\bold{z}) = \alpha \bold{I}_{N_{rs}}{\bold{z}}+{\bold{n_q}}$. The quantizer $\bold{Q}_b(\bold{z})$ is modeled as an Additive Quantization Noise Model (AQNM), where $\alpha = 1 - \frac{\pi\sqrt{3}}{2}2^{-2b}$, and $b$ is the bit resolution of the ADCs employed across all the RF paths \cite{Rangan,Uplink} in the receiver. The vector $\bold{n}_q$ is the additive quantization noise which is uncorrelated with $\bold{z}$ and has a Gaussian distribution: ${\bold{n}_q} \sim \mathcal{CN}(\bold{0},{\bold{D}_q^2})$ \cite{Rangan,Uplink}. This signal is later combined using the digital combiner ${\bold{W}_D^H}$ to produce the output signal ${\bold{y}} = {\bold{W}_D^H}{\bold{\tilde{y}}}$.

The relationship between the transmitted signal vector $\bold{x}$ and the received symbol vector $\bold{y}$ at the receiver is given by
\begin{equation}\label{eq_3}
\begin{split}
\bold{y} = \alpha\bold{W}_D^H\bold{W}_A^H\bold{P}\bold{\Phi}\bold{R}\bold{F}_A\bold{F}_D\bold{x} + \alpha \bold{W}_D^H\bold{W}_A^H \bold{n} + \bold{W}_D^H\bold{n_q}, 
\end{split}
\end{equation}
where the dimensions of the hybrid precoder and combiner are as follows:
$\bold{F}_D \in \mathbb{C}^{N_{rt} \times N}$, ${\bold{F}_A} \in \mathbb{C}^{N_t \times N_{rt}}$, ${\bold{W}_A^H} \in \mathbb{C}^{N_{rs} \times N_r}$, and ${\bold{W}_D^H} \in \mathbb{C}^{N \times N_{rs}}$.\\
\indent The precoders ${\bold{F}_D}$ and ${\bold{F}_A}$, and combiners ${\bold{W}_D^H}$ and ${\bold{W}_A^H}$ are designed for a given channel realization $\bold{H}$. We assume that the perfect channel state information $\bold{P}$ and $\bold{R}$ are known both to the transmitter and the receiver, and the topic of channel estimation is outside the scope of this paper. We further assume that the number of RF paths $N_{rs}$ on the receiver is the same as the number of parallel data streams $N$. The analysis is easy to extend and similar for the case $N_{rs} \ne N$.\\

\noindent \textit{Mean squared error performance:}\\

It can be shown that the expression for the MSE $\delta$ of the received, quantized, and combined signal $\bold{y}$ using \eqref{eq_3} as
\begin{equation}\label{eq_5}
\begin{split}
\delta &\triangleq \tr(\bold{M}(\bold{x})),
\end{split}
\end{equation}
where $\bold{M}(\bold{x})$ is the MSE matrix that can be written as
\begin{equation}\label{eq_5a}
\begin{split}
&\bold{M}(\bold{x}) = (E\big[ (\bold{y}-\bold{x})(\bold{y}-\bold{x})^H\big]),\\
&=  p(\bold{K}-\bold{I}_{N})(\bold{K}-\bold{I}_{N})^H + \alpha^2\sigma_{n}^2\bold{W}\bold{W}^H + \bold{W}_D^H\bold{D}_q^2\bold{W}_D.
\end{split}
\end{equation}
Here ${\bold{K}} = \alpha \bold{W}_D^H \bold{W}_A^H \bold{P}\bold{\Phi}\bold{R}\bold{F}_A\bold{F}_D , E[{\bold{x}}{\bold{x}}^H] = p{\bold{I}_{N}}, \bold{W} = \bold{W}_D^H \bold{W}_A^H, E[{\bold{n}}{\bold{n}}^H] = {\sigma_n^2}{\bold{I}_{N_r}}, E[{\bold{n_q}}{\bold{n_q}}^H] = \bold{D}_q^2, \bold{D}_q^2 = \alpha(1-\alpha){\text{diag}}[ {\bold{W}_A^H}{\bold{H}}({\bold{W}_A^H}{\bold{H}})^H+{\bold{I}_{N_{rs}}}], E[{\bold{n}}{\bold{n_q}}^H] = 0$, and $p$ is the average power of the symbol $\bold{x}$.\\

\indent The design of the precoder, combiner, and the RIS phase-shift settings to minimize the MSE $\delta$ for a given $b$-bit ADC can be posed as a multi-dimensional optimization problem
\begin{equation}\label{eq_6}
\begin{split}
&(\bold{F}_A^{opt},\bold{F}_D^{opt},{\bold{W}_A^H}^{opt},{\bold{W}_D^H}^{opt},\bold{\Phi}^{opt}) = \argminF_{\bold{F}_A,\bold{F}_D,{\bold{W}_A^H},{\bold{W}_D^H},\bold{\Phi}}\delta.
\end{split}
\end{equation}
If the precoders, combiners, and the RIS phase settings are chosen such that $\bold{K} = \bold{I}_{N}$, then the MSE matrix $\bold{M}(\bold{x})$ can be written as
\begin{equation}\label{eq_6a}
\bold{M}(\bold{x}) = \alpha^2\sigma_{n}^2\bold{W}\bold{W}^H + \bold{W}_D^H\bold{D}_q^2\bold{W}_D.
\end{equation}
An alternate equivalent problem to \eqref{eq_6} can be posed as
\begin{equation}\label{eq_7}
\begin{split}
\bold{K} = \alpha \bold{W}_D^H \bold{W}_A^H \bold{P}\bold{\Phi}\bold{R}\bold{F}_A\bold{F}_D &= \bold{I}_{N},\\
\text{ such that } \alpha^2\sigma_{n}^2\bold{W}\bold{W}^H + \bold{W}_D^H\bold{D}_q^2\bold{W}_D &= \bold{0}.
\end{split}
\end{equation}
Both \eqref{eq_6} and \eqref{eq_7} are challenging to solve given the constraints on the analog precoder and combiner \cite{Zakir7}. We take a multi-step approach to solve the problem by designing the hybrid precoder and combiner as a first step. In the next step, we derive the RIS phase setting, followed by fine-tuning the design of the digital precoder and combiner.

\section{Precoder and combiner Design}\label{PrecCombDesgn}
In order to design the precoders and combiners, we factor the digital precoder and combiner as
\begin{equation}\label{eq_8}
\bold{F}_D = \bold{\tilde{F}}_D\bold{F}_S,\text{ }\bold{W}_D^H = \bold{W}_S\bold{\tilde{W}}_D^H,
\end{equation}
where $\bold{\tilde{F}}_D \in \mathbb{C}^{N \times M}, \bold{F}_S \in \mathbb{C}^{M \times N}, \bold{\tilde{W}}_D^H \in \mathbb{C}^{M \times N}$, and $\bold{W}_S \in \mathbb{C}^{N \times M}$. This is illustrated using Fig. \ref{fig2}. We first focus on designing the partial digital precoder $\bold{\tilde{F}}_D$ and partial digital combiner $\bold{\tilde{W}}_D^H$, and the analog precoder $\bold{F}_A$ and analog combiner $\bold{W}_A^H$. We will later revisit the design of the other component of the digital precoder and combiner $\bold{F}_S$ and $\bold{W}_S$ in section \ref{part_dsgn}.

The hybrid precoding and combing techniques for systems employing phase shifters in mmWave transceiver architectures impose constraints on them. The analog precoder $\bold{F}_A$ and combiner ${\bold{W}_A^H}$ entries need to satisfy unit norm entries in them \cite{SigProc,PreDsgn,Zakir2,Zakir7}. We design the analog precoder $\bold{F}_A$ and the partial digital precoder $\bold{\tilde{F}}_D$ such that
\begin{equation}\label{eq_9}
\begin{split}
\bold{R}\bold{F}_A\bold{\tilde{F}}_D \approx \bold{I}_M.
\end{split}
\end{equation}
The hybrid precoders are derived upon solving the optimization problem \cite{PreDsgn,SigProc} stated below. 
\begin{equation}\label{eq_10}
\begin{aligned}
({\bold{F}_A^{opt}},{\bold{\tilde{F}}_D^{opt}}) = & \argminF_{{\bold{\tilde{F}}_D},{\bold{F}_A}}{\lVert {\bold{R}^{\dagger} - {{\bold{F}_A}{\bold{\tilde{F}}_D}}} \rVert }_F, \\
\text{such that } & {\bold{F_A}}\in{\mathcal{F}_{RF}}, {\lVert {{\bold{\tilde{F}}_D}{\bold{F}_A}} \rVert }_F^2 = N.
\end{aligned}
\end{equation}
The set $\mathcal{F}_{RF}$ is the set of all possible analog precoders that correspond to a hybrid precoder architecture based on phase shifters. This includes all possible $N_t \times N_{rt}$ matrices with constant magnitude entries. The term $\bold{R}^{\dagger}$ denotes the right inverse of $\bold{R}$.\\
\indent Similarly, the analog combiner $\bold{W}_A^H$ and the partial digital combiner $\bold{\tilde{W}}_D^H$ are designed such that
\begin{equation}\label{eq_11}
\begin{split}
\bold{\tilde{W}}_D^H\bold{W}_A^H\bold{P} \approx \bold{I}_M.
\end{split}
\end{equation}
The hybrid combiners are derived using \cite{PreDsgn}
\begin{equation}\label{eq_12}
\begin{aligned}
({\bold{W}_A^H}^{opt},{\bold{\tilde{W}}_D}^{H^{opt}}) &= \argminF_{\bold{\tilde{W}_D^H},{\bold{W}_A^H}}{\lVert {\bold{P}^{\ddagger} - {\bold{\tilde{W}}_D^H}{{\bold{W}_A^H}}} \rVert }_F, \\
\text{such that } & {\bold{W}_A^H}\in{\mathcal{W}_{RF}}, {\lVert {{\bold{\tilde{W}}_D^H}{\bold{W}_A^H}} \rVert }_F^2 = N.
\end{aligned}
\end{equation}
Here again the set $\mathcal{W}_{RF}$ is the set of all possible analog combiners that correspond to hybrid combiner architecture based on phase shifters. This includes all possible $N_{rs} \times N_r$ matrices with constant magnitude entries. The term $\bold{P}^{\ddagger}$ denotes the left inverse of $\bold{P}$.

\section{RIS phase shift optimization}\label{ris_opt}
In this section, we derive the expression for the CRLB of the MSE of the received, quantized, and combined signal $\bold{y}$ for a fixed $\bold{W}_A^H$, $\bold{F}_A$, $\bold{\tilde{W}}_D^H$, $\bold{\tilde{F}}_D$, and ADC bit resolution $b$ on all the RF paths of the receiver, and show that the MSE achieves the CRLB. We later formulate an optimization problem to minimize the MSE (or CRLB) for RIS phase-shift setting. Finally, we describe a design to fine-tune the precoder $\bold{F}_S$ and the combiner $\bold{W}_S$ considering the optimal RIS phase-shift settings.

\subsection{CRLB as function of RIS phase-shift settings}\label{crlb}
Given the analog combiner $\bold{W}_A^H$, analog precoder $\bold{F}_A$, the partial digital combiner $\bold{\tilde{W}}_D^H$, and the partial digital precoder $\bold{\tilde{F}}_D$ are derived using \eqref{eq_10} and \eqref{eq_12}, we substitute them in \eqref{eq_3} and rewrite the same as
\begin{equation}\label{eq_13}
\begin{split}
\bold{y} = \bold{K}\bold{x} + \bold{n}_1,
\end{split}
\end{equation}
where $\bold{K} = \alpha\bold{W}_S\bold{\Phi}\bold{F}_S$, and $\bold{n}_1 = \alpha \bold{W}_S\bold{\tilde{W}}_D^H\bold{W}_A^H \bold{n} + \bold{W}_S\bold{\tilde{W}}_D^H\bold{n_q}$. We know that $\bold{n}$ and $\bold{n_q}$ are Gaussian random vectors such that $\bold{n} \sim \mathcal{N}(\bold{0},{\sigma_n^2}{\bold{I}_{N_r}})$ and $\bold{n_q} \sim \mathcal{N}(\bold{0},{\bold{D}_q^2})$ respectively.
%
Hence we have
\begin{equation}\label{eq_15}
E[\bold{n_1}] = \alpha \bold{W}_S\bold{\tilde{W}}_D^H\bold{W}_A^H E[{\bold{n}}] + \bold{W}_S\bold{\tilde{W}}_D^H E[{\bold{n_q}}] = {\bold{0}},
\end{equation}
\begin{equation}\label{eq18a}
\begin{split}
{\sigma_{n_1}^2} = E[{\bold{n_1}}{\bold{n_1}}^H] = \alpha^2\sigma_n^2\bold{W}\bold{W}^H +  \bold{W}_S\bold{\tilde{W}}_D^H\bold{D}_q^2\bold{\tilde{W}}_D\bold{W}_S^H.
\end{split}
\end{equation} 
Thus $\bold{n_1} \sim \mathcal{N}(\bold{0},(\alpha^2\sigma_n^2\bold{W}\bold{W}^H +  \bold{W}_S\bold{\tilde{W}}_D^H\bold{D}_q^2\bold{\tilde{W}}_D\bold{W}_S^H))$.
It is noted that ${\bold{W}}$ is an $N \times N_r$ matrix with $N_r \gg N$. It is safe to assume that $\bold{W}$ has a full row rank and its pseudo-inverse exists.
Equation \eqref{eq_13} can be seen as a linear model, in which we intend to estimate $\bold{x}$, given the observation $\bold{y}$. We can express the conditional probability distribution of ${\bold{y}}$ given ${\bold{x}}$ as
\begin{equation}\label{eq_17}
p({\bold{y} \vert {\bold{x}}}) \sim \frac{1}{({2\pi}{{\sigma_{n_1}^2}})^{\frac{N}{2}}} \text{exp} \bigg\{ -\frac{1}{2{\sigma_{n_1}^2}} ({\bold{y}}-{\bold{K}}{\bold{x}})^H({\bold{y}}-{\bold{K}}{\bold{x}}) \bigg\}.
\end{equation}
From \eqref{eq_13} and \eqref{eq_17}, it is straightforward to see that the ``regularity conditions" are satisfied, and hence for such a linear estimator, we can write the expression for the CRLB as
\begin{equation}\label{eq_19_a}
\begin{split}
&{\bold{I}^{-1}({\bold{\hat{x}}})} = ({\bold{K}^H}{\bold{C}^{-1}}{\bold{K}})^{-1}\\
&= \bold{F}_S^{-1}\Big[ \sigma_n^2\bold{\bold{\Phi}^{-1}}\bold{\tilde{W}}\bold{\Phi}  + \frac{1}{\alpha^2} \bold{\Phi}^{-1}\bold{\tilde{W}}_D^H\bold{D}_q^2\bold{\tilde{W}}_D\bold{\Phi} \Big](\bold{F}_S^H)^{-1},\\
\end{split}
\end{equation}
where $\bold{\tilde{W}} = \bold{\tilde{W}}_D^H\bold{W}_A^H\bold{W}_A\bold{\tilde{W}}_D$, $\bold{D}_q^2 = \alpha(1-\alpha)\diag\Big[ {\bold{(\tilde{W}}_D^H})^{-1}\bold{\Phi}\bold{R}\bold{R}^H\bold{\Phi}^{-1}\bold{\tilde{W}}_D^{-1} + \bold{I}_{N} \Big]$, and C the noise covariance matrix of ${\bold{n_1}}$. The details of the proof are given in Appendix \ref{AppA}.\\
\indent It can also be seen that if the precoders, combiners, and the phase shift settings are designed such that $\bold{K} = \bold{I}_{N}$, the MSE in \eqref{eq_6a} achieves the CRLB. Formally,
\begin{equation}\label{eq_19}
\begin{split}
{\bold{I}^{-1}({\bold{\hat{x}}})} &= \alpha^2\sigma_n^2\bold{W}\bold{W}^H +  \bold{W}_S\bold{\tilde{W}}_D^H\bold{D}_q^2\bold{\tilde{W}}_D\bold{W}_S^H,\\
&= \alpha^2\sigma_n^2\bold{W}\bold{W}^H + \bold{W}_D^H\bold{D}_q^2\bold{W}_D,\\
&= \bold{M}(\bold{x}).
\end{split}
\end{equation}

\subsection{Design of the RIS phase shift matrix}\label{ris_dsgn}
Minimizing the CRLB in \eqref{eq_19} will ensure the minimum MSE ($\delta$) performance for a given fixed $\bold{W}_A^H$, $\bold{F}_A$, $\bold{\tilde{W}}_D^H$, $\bold{\tilde{F}}_D$, and ADC bit resolution $b$. The CRLB \eqref{eq_19} can be minimized when
\begin{equation}\label{eq_19b}
\begin{split}
\bold{\Phi}^{-1}\bold{\tilde{W}}_D^H\Big[ \sigma_n^2\bold{W}_A^H\bold{W}_A + \frac{1}{\alpha^2}\bold{D}_q^2 \Big]\bold{\tilde{W}}_D\bold{\Phi} = \bold{0}.
\end{split}
\end{equation}
Thus the design of the RIS phase shift matrix can be posed as
\begin{equation}\label{eq_21}
\begin{split}
\bold{\Phi}^{opt} = \argminF_{\bold{\Phi}}{\lVert \bold{\Phi}^{-1}\bold{\tilde{W}}_D^H\Big[ \sigma_n^2\bold{W}_A^H\bold{W}_A + \frac{1}{\alpha^2}\bold{D}_q^2 \Big]\bold{\tilde{W}}_D\bold{\Phi} \rVert }_F^2.
\end{split}
\end{equation}
It can also be shown further that minimizing \eqref{eq_19} is equivalent to maximizing the throughput, and energy-efficiency of the wireless link. Please refer to Appendix \ref{AppB} for the proofs. 

\subsection{Design of the partial digital precoder and combiner}\label{part_dsgn}
Now we revisit the design of the other partial digital precoder $\bold{F}_S$ and combiner $\bold{W}_S$. By substituting all the designed parameters into \eqref{eq_7}, we have
\begin{equation}\label{eq_23}
\begin{split}
\bold{K} = \bold{W}_S\bold{\Phi}^{opt}\bold{F}_S = \bold{I}_M.
\end{split}
\end{equation}
By appropriately selecting a matrix $\bold{F}_S^{opt}$ such that its right inverse $({\bold{F}_S^{opt}})^{\dagger}$ exists we can rewrite \eqref{eq_23} as
\begin{equation}\label{eq_24}
\begin{split}
{\bold{W}_S^{opt}} = ({\bold{F}_S^{opt}})^{\dagger}(\bold{\Phi}^{opt})^{-1}.
\end{split}
\end{equation}
It is to be noted that $(\bold{\Phi}^{opt})^{H} = (\bold{\Phi}^{opt})^{-1}$ and the inverse $(\bold{\Phi}^{opt})^{-1}$ always exists. In the next section we describe an algorithm to solve \eqref{eq_21}.

\section{Information-directed branch-and-prune Algorithm}\label{idbp_algo}
In this section, we detail the proposed IDBP algorithm, which belongs to the family of tree-traversal search methods. The IDBP is very different compared from the existing tree-based algorithms. In it, we use an information-theoretic measure to decide on the branching and pruning rules during tree traversal to arrive at an optimal sequence or solution in probability. It is also vastly different compared to the well-known branch-and-bound algorithms \cite{Bbound,BoydBnB}. The difference between the two mainly is that in IDBP, the pruning decisions are not based on bounds of the reward or cost of the optimal solution, instead, they are based on an information-theoretic measure. For the first time in the literature, we provide theoretical guarantees for near-optimality with the proposed IDBP algorithm using asymptotic equipartition theory.

\subsection{The problem setup}\label{InfMeas}
We model the solution $\bold{\Phi}$ as a sequence of random variables $\mathit{\Phi} = \{ \Phi_1, \Phi_2, \cdots, \Phi_M \}$, where we represent the discrete random variable $\Phi_i$ with probability mass function (PMF) $p(\Phi_i)$. The solution can be visualized as a finite horizon Markov decision process (MDP), which is defined using a tuple $(\Phi,\mathcal{A},p,r,q)$, where $\Phi$ denotes the finite set of states, $\mathcal{A}$ is the finite set of actions, $p:\Phi \times \mathcal{A} \times \Phi^{'} \rightarrow [0,1]$ are the state transition probabilities $p_{\phi,a}(\phi^{'})$ that a state $\phi^{'}$ is attained when an action $a \in \mathcal{A}$ is taken in state $\phi$ where $\phi,\phi^{'} \in \Phi$. A reward $r:\Phi \times \Phi \rightarrow \mathbb{R}$ is associated with an $a \in \mathcal{A}$ from a state $\phi \in \Phi$. The prior distribution $q$ represents the statistic of the optimal solution. We consider the actions $a \in \mathcal{A}$ to be deterministic given $p$ and $q$. We define a solution $\pi = \{\Phi_1=\phi_1, \Phi_2=\phi_2, \cdots, \Phi_M=\phi_M \}$ as a sequence of states attained as a consequence of decisions $a \in \mathcal{A}$ taken to maximize the cumulative reward in the MDP. The details about the MDP assumption for the solution to \eqref{eq_21} is detailed in Appendix \ref{AppC}. \\
\indent 
We design the IDBP algorithm with pruning rules so as to minimize the effective Kullback-Leibler (KL) divergence between the distribution of future looking sequence $\{ \Phi_{m+1}, \cdots, \Phi_M \}$ given $\Phi_m$ with respect to the known prior conditional distribution of the successive future states $q(\Phi_{m+1}, \Phi_{m+2}, \cdots, \Phi_M | \Phi_m)$. We call this algorithm the IDBP. Effectively, we can write \cite{Tishby_RL}
\begin{equation}\label{eq_bp_1}
\begin{split}
\pi^{opt} = \argminF_{\pi}\{D_{KL}(p(\Phi_1,\cdots,\Phi_{M})||q(\Phi_1,\cdots,\Phi_{M}))\}.
\end{split}
\end{equation}
Using the Asymptotic Equipartition Property (AEP), it can be shown that the solution $\pi^{opt}$ is optimal in probability. This is detailed in the next subsection. We say that the solution $\pi^{opt}$ is close to $\pi^*$ in probability when $Pr\big\{ \left\vert\ f(\pi^{opt}) - f(\pi^*) \right\vert \le \epsilon \big\} \ge 1-\delta$, where $\epsilon, \delta$ can be chosen arbitrarily close to zero. Here, $\pi^*$ is the optimal solution to \eqref{eq_21}. That is $\bold{\Phi}^{opt} = \diag(\pi^*)$. Here $f(\cdot)$ is the objective function of \eqref{eq_21}, defined in \eqref{eq_bp_4}.

\subsection{Optimality Analysis}\label{OptAnalysis}
In this subsection, we provide the proofs of Theorems \ref{thm1} - \ref{thm3}, and Lemma \ref{lemm-1} that establishes theoretical guarantees of the optimal solution in probability using the proposed IDBP Algorithm. Before laying out the details of the optimality analysis, we first describe one of the methods that can be used to derive the statistics $q$.

\subsubsection{Determination of the priors of the optimal solution $q$}\label{qeval}
Given that we model the solution as an MDP, we write the statistics of the optimal solution $\pi^*$ as $\pi^* \sim q(\Phi_1,\Phi_2,\cdots,\Phi_M)$, where $q(\Phi_1,\Phi_2,\cdots,\Phi_M) = q(\Phi_1)q(\Phi_2| \Phi_1) \cdots q(\Phi_M| \Phi_{M-1})$. Here $q(\Phi_1)$ is initial state distribution. We assume that the MDP is homogenous and hence it is sufficient to determine the transition probabilities $q(\Phi_{t+1} = \phi_i | \Phi_t = \phi_j)$ between any two consecutive stages $t$ and $t+1, \forall t \in [1,M); \phi_i, \phi_j \in \Phi$. To do so, we identify $m$ solutions $\{ \pi^i \}_{i=1}^{m}$ from the exhaustive search space of problem \eqref{eq_21} such that $f(\pi^1) \le f(\pi^2) \le \cdots \le f(\pi^{m})$, where $f(\cdot)$ is the objective of the problem \eqref{eq_21} written as $f(\pi) = $
\begin{equation}\label{eq_bp_4}
\begin{split}
{\Big\lVert {\diag(\pi)}^{-1}\bold{\tilde{W}}_D^H\Big[ \sigma_n^2\bold{W}_A^H\bold{W}_A + \frac{1}{\alpha^2}\bold{D}_q^2 \Big]\bold{\tilde{W}}_D\diag(\pi) \Big\rVert }_F^2.
\end{split}
\end{equation}
Using these $m$ subset of solutions we evaluate
\begin{equation}\label{eq_bp_5}
\begin{split}
q(\Phi_{t+1} &= \phi_i | \Phi_{t}=\phi_j) = \frac{F(\{ \Phi_{t+1}=\phi_i | \Phi_{t}=\phi_j \})}{mM}\\
&\forall t \in [1,M) ; \phi_i, \phi_j \in \Phi,
\end{split}
\end{equation}
Here $F(\{ \Phi_{t+1}=\phi_i | \Phi_{t}=\phi_j \})$ returns the number of times the event $\{ \Phi_{t+1}=\phi_i | \Phi_{t}=\phi_j \}$ occur among the $m$ solutions. It follows that if $\pi^* \in \{ \pi^i \}_{i=1}^{m}$, and for a small $m$ we have $q(\pi^*) \rightarrow 1$.\\
Alternatively, one can also use other fast non-parametric techniques or heuristic approaches to estimate the conditional priors $q$ \cite{qevalm,qevalNN}.\\
\indent We know that MDP $\mathit{\Phi} = \{ \Phi_1, \Phi_2, \cdots, \Phi_M \}$ can be visualized as homogenous Markov source, and exhibits asymptotic equipartition property. In addition, we also know the following \cite{Thomas, StrTyp}
\newtheorem{definition}{Definition}
\begin{definition}\label{def1}
A sequence $\pi_n$ (or a solution of length $n$) is strongly $\delta$ typical with respect to the distribution $q$ if $\forall \phi \in \Phi : |q_{\pi_n}(\phi) - q(\phi) | \le \delta q(\phi)$.\\
\end{definition}

Here $q_{\pi_n}(\phi) = \frac{\ell(\phi)}{n}$ is the empirical distribution signifying the number of occurrences of $\phi$ denoted as $\ell(\phi)$ over $n$ observations.
\begin{definition}\label{def2}
The strongly $\delta$-typical set, $\mathcal{T}_{\delta}^n(\Phi)$ is a set of all strongly $\delta$ typical sequences. That is 
\begin{equation}\label{eq_def2}
\begin{split}
\mathcal{T}_{\delta}^n(\Phi) = \Big\{ \pi_n : \Big| q_{\pi_n}(\phi) - q(\phi) \Big| \le \delta q(\phi) \Big\}.
\end{split}
\end{equation}
\end{definition}

\begin{definition}\label{def3}
The weakly $\epsilon$-typical set, $A_{\epsilon}^n(\Phi)$ is a set of all sequences such that
\begin{equation}\label{eq_def3}
\begin{split}
A_{\epsilon}^n(\Phi) = \Big\{ \pi_n : \Big| -\frac{1}{n} \log q(\pi_n) - H(\Phi) \Big| \le \epsilon \Big\},
\end{split}
\end{equation}
\end{definition}
where $H(\Phi)$ is the source entropy rate of the MDP under consideration.
We now modify the Definition \ref{def1} to incorporate the conditional priors $q(\Phi_{t+1} = \phi_i | \Phi_{t}=\phi_j)$, and show that the solution $\pi_n$ belongs to $A_{\eta}^n(\Phi)$, for some $\eta \to 0$, when the following condition $|q_{\pi_n}(\phi_i|\phi_j) - q(\phi_i|\phi_j) | \le \delta q(\phi_i|\phi_j)$ is satisfied. 
\newtheorem{theorem}{Theorem}
\begin{theorem}\label{thm1}
A sequence $\pi_n$ is $\eta$ typical with respect to the conditional distribution $q$ if $\forall \phi_i, \phi_j \in \Phi : |q_{\pi_n}(\phi_i|\phi_j) - q(\phi_i|\phi_j) | \le \delta q(\phi_i|\phi_j)$, for some $\eta, \delta \to 0$.\\
\end{theorem}
\begin{proof}
We have $q_{\pi_n}(\phi_i|\phi_j)$ the empirical conditional distribution of the sequence $\pi_n$ defined as
\begin{equation}\label{eq_bp_5_1}
\begin{split}
q_{\pi_n}(\Phi_{t+1} &= \phi_i | \Phi_{t}=\phi_j) = q_{\pi_n}(\phi_i | \phi_j) = \frac{\ell(\{ \phi_i | \phi_{j}\};\pi_n)}{n},\\
&\forall t \in [1,n) ; \phi_i, \phi_j \in \Phi,
\end{split}
\end{equation}
where $\ell(\{ \phi_i | \phi_{j}\};\pi_n)$ denotes the number of occurrences of the transitions $\phi_i$ to $\phi_j$ in the sequence $\pi_n$. Let the sequence $\pi_n = \{ \phi_{t(1)}, \phi_{t(2)}, \cdots, \phi_{t(n)} \}$, where $\phi_{t(i)} \in \Phi, \forall i \in [1,n]$. Then we have 
\begin{equation}\label{eq_bp_6}
\begin{split}
q(\pi_n) &= q(\phi_{t(1)})^{\ell(\phi_{t(1)};\pi_n)} \prod_{i=2}^{n-1} q(\phi_{t(i+1)} | \phi_{t(i)})^{\ell(\phi_{t(i+1)} | \phi_{t(i)};\pi_n)},
\end{split}
\end{equation}
where $\ell(\phi_{t(1)};\pi_n)$ is the number of occurrences of the state $\phi_{t(1)}$ in the sequence (solution) $\pi_n$. We write \eqref{eq_bp_6} as
\begin{equation}\label{eq_bp_7}
\begin{split}
\log(q(\pi_n)) &= \ell(\phi_{t(1)};\pi_n)\log q(\phi_{t(1)})\\
&+ \sum_{i=2}^{n-1} \ell(\{ \phi_{t(i+1)} | \phi_{t(i)} \};\pi_n) \log q(\phi_{t(i+1)} | \phi_{t(i)})
\end{split}
\end{equation}
For simplicity of notation, we represent the conditionals $\{ \phi_{t(i+1)} | \phi_{t(i)} \}$ as $\psi_{i}$, $\phi_{t(1)}$ as $\psi_1$, and $\ell(\phi_{t(1)};\pi_n)$ as $\ell(\psi_1)$  We simplify \eqref{eq_bp_7} further as
\begin{equation}\label{eq_bp_8}
\begin{split}
\log q(\pi_n) &= \sum_{i=1}^{n-1} \ell(\psi_i) \log q(\psi_i),\\
&= \sum_{i=1}^{n-1} \Big\{ \ell(\psi_i) -nq(\psi_i) +nq(\psi_i) \Big\}\log q(\psi_i),\\
&= n \sum_{i=1}^{n-1} q(\psi_i)\log q(\psi_i)\\
&+ n \sum_{i=1}^{n-1} \Big( \frac{1}{n}\ell(\psi_i) - q(\psi_i) \Big) \log q(\psi_i),\\
&= -n \{ H(\Phi) + \eta \}
\end{split}
\end{equation}
where $H(\Phi) = H(\Phi_1) + \sum_{i=2}^{N-1} H(\Phi_{i+1}|\Phi_i)$ for the MDP under consideration \cite{Thomas}, and 
\begin{equation}\label{eq_bp_9}
\begin{split}
\eta &= \sum_{i=1}^{n-1} \Big( \frac{1}{n}\ell(\psi_i) - q(\psi_i) \Big)(-\log q(\psi_i)),\\
&\le \sum_{i=1}^{n-1} \Big| \frac{1}{n}\ell(\psi_i) - q(\psi_i) \Big| (-\log q(\psi_i)).
\end{split}
\end{equation}
We know that $\Big| \frac{1}{n}\ell(\psi_i) - q(\psi_i) \Big| = \Big| q_{\pi_n}(\phi_i|\phi_j) - q(\phi_i|\phi_j) \Big|$ for $i \in [1,n]$, and hence we have
%
\begin{equation}\label{eq_bp_10}
\begin{split}
\eta &\le \delta \sum_{i=1}^{n-1} q(\psi_i) (-\log q(\psi_i)) = \Big| \delta H(\Phi) \Big|,\text{ or}\\
&\le \hat{\eta}, \text{ where }\hat{\eta} = \Big| \delta H(\Phi) \Big|.
\end{split}
\end{equation}
It is straightforward to see that for a finite $N$, $\hat{\eta} \to 0$ as $\delta \to 0$. Hence we can write \eqref{eq_bp_9} as
\begin{equation}\label{eq_bp_11}
\begin{split}
\Big(H(\Phi) - \hat{\eta} \Big) \le -\frac{1}{n} \log q(\pi_n) \le \Big(H(\Phi) + \hat{\eta} \Big)
\end{split}
\end{equation}
\end{proof} 
We now show that the optimal sequence $\pi^* \in A_{\epsilon}^M(\Phi)$.

\begin{theorem}\label{thm2}
Let $q$ be the conditional priors derived using the $m$-best sequences $\{ \pi^i \}_{i=1}^m$ as described in \eqref{eq_bp_5} that accurately represent the optimal solution $\pi^*$, then $\pi^* \in A_{\epsilon}^M(\Phi)$.
\end{theorem}
\begin{proof}
Let the $m$-best sequences be denoted as
\begin{equation}\label{eq_T2_1}
\begin{split}
\pi^i = \{ \phi_1^i, \phi_2^i, \cdots, \phi_M^i \},\text{ where }\phi_j^i \in \Phi, \forall j \in [1,M].
\end{split}
\end{equation}
we now have the empirical distribution of the sequences as 
\begin{equation}\label{eq_T2_2}
\begin{split}
&\hat{q}(\pi^i) = \hat{q}_{\pi^i}(\phi_{1}^i) \prod_{j=2}^{M-1} \hat{q}_{\pi^i}(\phi_{j+1}|\phi_{j}),\\
&\text{where }\hat{q}_{\pi^i}(\phi_{1}^i) = \frac{\ell(\phi_{1}^i; \pi^i)}{M} = \frac{\ell(\{ \phi_{1}^i|\phi_{0} \}; \pi^i)}{M},\\
&\hat{q}_{\pi^i}(\phi_{j+1}^i|\phi_{j}^i) = \frac{\ell (\{ \phi_{j+1}^i|\phi_{j}^i \}; \pi^i)}{M-1}.
\end{split}
\end{equation}
It is also worth noting that the starting transition $\phi_{0}$ to $\phi_{1}^i$ occurs only once in the sequence. Hence in general we can write 
\begin{equation}\label{eq_T2_3}
\begin{split}
\hat{q}_{\pi^i}(\phi_{j+1}^i|\phi_{j}^i) = \frac{\ell (\{ \phi_{j+1}^i|\phi_{j}^i \}; \pi^i)}{M}.
\end{split}
\end{equation}
However from \eqref{eq_bp_5} we have
\begin{equation}\label{eq_T2_4}
\begin{split}
q(\phi_{j+1}|\phi_{j}) = \frac{F(\phi_{j+1}|\phi_{j})}{mM}.
\end{split}
\end{equation}
Since $F(\phi_{j+1}|\phi_{j})$ is the number of occurrences of the transitions $\phi_{j}$ to $\phi_{j+1}$ in all the $m$ sequences, we can rewrite \eqref{eq_T2_3} as
\begin{equation}\label{eq_T2_5}
\begin{split}
\sum_{i=1}^{m} \hat{q}_{\pi^i}(\phi_{j+1}^i|\phi_{j}^i) &= \frac{1}{M} \sum_{i=1}^m \ell (\phi_{j+1}^i|\phi_{j}^i; \pi^i) = \frac{1}{M} F(\phi_{j+1}|\phi_{j}).
\end{split}
\end{equation}
We say that the empirical priors $\hat{q}$ is an accurate representation of the optimal sequence $\pi^*$ if 
\begin{equation}\label{eq_T2_6}
\begin{split}
\hat{q}_{\pi^1}(\phi_{j+1}|\phi_{j}) &\approx \hat{q}_{\pi^2}(\phi_{j+1}|\phi_{j}) \approx \cdots \approx \hat{q}_{\pi^m}(\phi_{j+1}|\phi_{j})\\
&\approx \hat{q}_{\pi^*}(\phi_{j+1}|\phi_{j}) \forall j \in [1,M-1]. 
\end{split}
\end{equation}
substituting \eqref{eq_T2_6} in \eqref{eq_T2_5} we have
\begin{equation}\nonumber
\begin{split}
m \hat{q}_{\pi^*}(\phi_{j+1}|\phi_{j}) & \approx \frac{1}{M} F(\phi_{j+1}|\phi_{j}),\\
\hat{q}_{\pi^*}(\phi_{j+1}|\phi_{j}) & \approx \frac{1}{mM} F(\phi_{j+1}|\phi_{j}),
\end{split}
\end{equation}
\begin{equation}\label{eq_T2_7}
\begin{split}
\hat{q}_{\pi^*}(\phi_{j+1}|\phi_{j}) \approx q(\phi_{j+1}|\phi_{j}), \forall j \in [1,M-1].
\end{split}
\end{equation}
From \eqref{eq_T2_7} we can write
\begin{equation}\label{eq_T2_8}
\begin{split}
&\Big| \hat{q}_{\pi^*}(\phi_{j+1}|\phi_{j}) - q(\phi_{j+1}|\phi_{j}) \Big| \le \delta q(\phi_{j+1}|\phi_{j}),\\
&\forall j \in [1,M-1], \text{ and for some }\delta \to 0.
\end{split}
\end{equation}
Now using Theorem \ref{thm1} we can write $\pi^* \in A_{\epsilon}^M(\Phi)$ w.r.t conditional $q$; if the statistic $q$ is a close representation of the optimal solution $\pi^*$.
\end{proof}
Using the proposed IDBP algorithm we find another sequence $\pi^p$ as a solution, drawn from a conditional distribution $p(\phi_i | \phi_j), \forall t \in [1,M) ; \phi_i, \phi_j \in \Phi$ such that $p(\pi^p) \approx q(\pi^*)$, and $D_{KL}(p||q) \to 0$. We now show that the sequences $\pi^p, \pi^* \in A_{\eta}^n(\Phi)$ w.r.t the conditional $q$ for some $\eta \to 0$.

\begin{theorem}\label{thm3}
Let $\pi^p$ be a sequence obtained using the conditional distribution $p(\phi_i|\phi_j)$ such that $p_{\pi^p}(\phi_i|\phi_j) \approx q_{\pi^*}(\phi_i|\phi_j), \phi_i, \phi_j \in \Phi$, and $D_{KL}(p||q) \to 0$; then it can be shown that the sequence $\pi^p$ and $\pi^*$ belong to the typical set w.r.t the conditional $q$. That is $\pi^p, \pi^* \in A_{\eta}^M(\Phi)$ for some $\eta \to 0$.
\end{theorem}
\begin{proof}
We have
\begin{equation}\label{eq_T3_1}
\begin{split}
&p_{\pi^p}(\phi_i|\phi_j) \approx q_{\pi^*}(\phi_i|\phi_j), \forall \phi_i, \phi_j \in \Phi, D_{KL}(p||q) \to 0.
\end{split}
\end{equation}
From Theorem \ref{thm2}, we have $\pi^* \in A_{\epsilon}^M(\Phi)$, and using Theorem \ref{thm1} we can write
\begin{equation}\label{eq_T3_2}
\begin{split}
&\Big| q_{\pi^*}(\phi_i|\phi_j) - q(\phi_i|\phi_j) \Big| \le \delta q(\phi_i|\phi_j),\text{ or}\\
&\Big| p_{\pi^p}(\phi_i|\phi_j) - q(\phi_i|\phi_j) \Big| \le \delta^{'} q(\phi_i|\phi_j).\text{ (using \eqref{eq_T3_1})}
\end{split}
\end{equation}
where $\delta^{'} \to 0$. For some $\eta = \max(\delta, \delta^{'})$, we can write the following 
\begin{equation}\label{eq_T3_4}
\begin{split}
\Big| p_{\pi^p}(\phi_i|\phi_j) - q(\phi_i|\phi_j) \Big| &\le \eta q(\phi_i|\phi_j),\\
\Big| q_{\pi^*}(\phi_i|\phi_j) - q(\phi_i|\phi_j) \Big| &\le \eta q(\phi_i|\phi_j),
\end{split}
\end{equation}
where $\eta \to 0$, and $\forall \phi_i, \phi_j \in \Phi$. Hence we have $\pi^p, \pi^* \in A_{\eta}^M(\Phi)$.
\end{proof}
Finally, it follows that if $\pi^p \in A_{\eta}^M(\Phi)$ w.r.t the conditional $q$, which is a close representation of $\pi^*$, then $\pi^p$ is optimal solution in probability.

\newtheorem{Lemma}{Lemma}
\begin{Lemma}\label{lemm-1}
If $\pi^p, \pi^* \in A_{\eta}^M(\Phi)$ w.r.t the conditionals $q$, and if $q$ is a close representation of the optimal solution $\pi^*$ we have $Pr\Big\{ \left\vert\ f(\pi^{opt}) - f(\pi^*) \right\vert \le \epsilon \Big\} \ge 1-\delta$, where $\epsilon$, $\delta$ are very small numbers not related to $\eta$.
\end{Lemma}
\begin{proof}
Since we have $\pi^p, \pi^* \in A_{\eta}^M(\Phi)$, we have
\begin{equation}\label{eq_T4_1}
\begin{split}
&p(\pi^p) = p(\pi^*) \approx 1,\text{ or }\pi^p \to \pi^*; \text { Hence we can safely write}\\
&Pr\Big\{ \left\vert\ f(\pi^{opt}) - f(\pi^*) \right\vert \le \epsilon \Big\} \ge 1-\delta.
\end{split}
\end{equation}
\end{proof}

\section{Algorithm Description}\label{AlgoDes}
\begin{figure}[t!]
\centering
\includegraphics[scale=0.28]{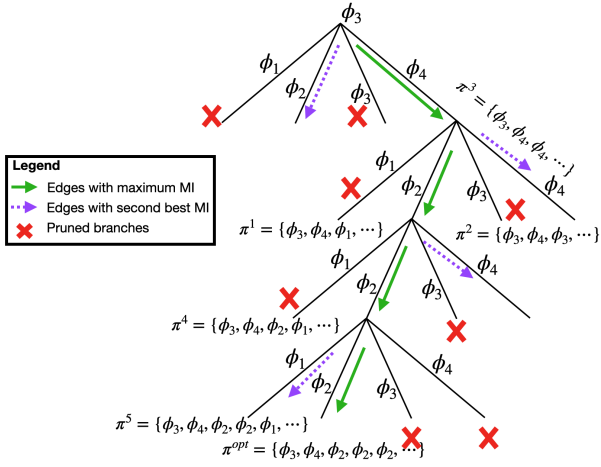}
\caption{\small An illustration of the tree traversal using the proposed IDBP Algorithm.}
\label{fig2_bp}
\end{figure}
Inspired by the well-known Chow-Liu Algorithm (CLA), we develop the proposed IDBP algorithm to arrive at the solution $\pi^{opt}$ \cite{CLAlgo}. The CLA minimizes the KL divergence between the actual distribution represented using the conditional priors $q$ and the distribution of $\pi^{opt}$. It finds the best second-order product approximation of the multi-dimensional discrete probability distribution from a finite set of observed data. The CLA finds the optimal tree-structured network $T(X_1,X_2,\cdots,X_k)$ of depth $k$ by minimizing the KL divergence between the observed (actual) distribution $p_t(X_1,X_2,\cdots, X_k)$ and the tree-structured distribution $T(X_1,X_2,\cdots,X_k)$. That is 
\begin{equation}\label{eq_bp_2}
\begin{split}
\min_{T}\{D_{KL}(p_t(X_1,\cdots, X_k)||T(X_1,\cdots,X_k))\},
\end{split}
\end{equation}
where $\{ X_1,X_2,\cdots, X_k \}$ is a sequence of random variables.
One of the key results from \cite{CLAlgo} is that, for minimizing the KL divergence in \eqref{eq_bp_2}, it is sufficient to find a tree network $T$ such that we maximize the mutual information (MI) $I(X_i,X_{\gamma(i)})$ between the tree edges in $T$. Here $X_{\gamma(i)}$ denotes the parent of $X_i$ in the tree under consideration.\\
\indent The proposed IDBP algorithm maximizes the MI between the tree edges (branches) to select the optimal-path edges and prune others. This ensures optimal solution in probability to \eqref{eq_21}. It is also worth noting that since we have an MDP model for our solution, it suffices to consider a second-order approximation for the joint probability distribution. Given the prior statistics $q$ of the optimal solution, and the transition probabilities $p$ between the phase settings, we traverse the tree by maintaining the edges that maximize the MI $I(X_i, X_{\gamma(i)})$. The proposed IDBP algorithm is described using Algorithm \ref{Algo1}. The algorithm yields an optimal solution in probability $\pi^{opt}$ if the priors $q$ selected is a close representation of the optimal solution $\pi^*$. In such a situation, the proposed IDBP algorithm requires a single pass tree traversal to get to the solution $\pi^{opt}$. This is the best case. However, in situations when $q$ is not an accurate representation of $\pi^*$, we propose to use a second pass from every node visited to traverse the tree along with the second-best child. This is described in Algorithm \ref{Algo1}. One can choose to extend the algorithm to explore $k$-best children. An illustration of the proposed IDBP tree search is shown in Fig. \ref{fig2_bp}. However, when extended to all the children, the algorithm becomes an exhaustive search. The process of designing the hybrid precoder, hybrid combiner, and the RIS phase configuration is outlined as design flow in Algorithm \ref{desflow}.
\begin{minipage}{0.48\textwidth}
\renewcommand*\footnoterule{}
\begin{savenotes}
\begin{algorithm}[H]
  \caption{Design flow}\label{desflow}
  \begin{algorithmic}[1]
   \small
      \Procedure{Design flow}{}
      	 \State $\{{\bold{F}_A^{opt}},{\bold{\tilde{F}}_D^{opt}}\}  \gets \text{ using \eqref{eq_10}}$
	 \State $\{{{\bold{W}_A^H}^{opt}},{\bold{\tilde{W}}_D}^{H^{opt}}\}  \gets \text{ using \eqref{eq_12} }$
	 \State $\bold{\Phi}^{opt}  \gets \text{ by solving \eqref{eq_21} using IDBP}$
	 \State $\{({\bold{F}_S^{opt}}), ({\bold{W}_S^{opt}})\}  \gets \text{ using \eqref{eq_23} and \eqref{eq_24} }$
	 \State{\textbf{return} ${\{ \bold{F}_A^{opt}},{\bold{W}_D^{opt}},\bold{\Phi}^{opt},\bold{\Phi}^{opt},{\bold{F}_S^{opt}},{\bold{W}_S^{opt}} \}$}
  \EndProcedure
  \end{algorithmic}
\end{algorithm}
%
\begin{algorithm}[H]
  \caption{Proposed IDBP}\label{Algo1}
  \begin{algorithmic}[1]
   \small
      \Function{IDBP}{$\Phi$,$M$,$m$}
      	 \State $\Phi  \gets \text{Finite set of phase angles with cardinality }K$
         \State $M  \gets \text{Number of RIS elements}$
      	\State $m \gets \text{ Number of sequences used to derive the priors }q$
        \State{$\text { InitializeStack() }$}
        \State{$\pi^{opt} \gets \emptyset; C_{opt} \gets \infty$}
        \State{$q \gets \text{ Compute the priors as described in Section \ref{qeval}}$}
        \State{$p \gets \text{ Compute the initial state probabilities}$}
        \State{$X_0 \gets \text{ Compute using $p$ and $q$}$}
         \State{$c \gets \text{ Compute initial cost using $p$ and $q$}$}
       \State{$\pi^{opt} \gets \text{ TraverseTree ($X_0$,$c$,$p$,$q$,$M$,2,1) }$}
       \State{\textbf{return} $\{ \pi^{opt} \}$} \Comment{Solution}  
       \EndFunction
       %
       %
       \Function{TraverseTree}{$X_{\text{curr}}$,$c$,$p$,$q$,$M$,stage,rec}
                \State $X_{\text{curr}}  \gets \text{ Current node in the tree}$
      		\State $r \gets \text{ Accumulated cost up till the node } X_{\text{curr}}$
      		\State{$q \gets \text{ The conditional priors}$}
      		\State{$p \gets \text{ The transition probabilities }$}
      		\State $M  \gets \text{Number of RIS elements}$
      		\State $\text{stage }  \gets \text{The current stage(level) in the tree traversal}$
      		\State $\text{rec }  \gets \text{Indicator to control recursion}$
    		\If { $\text{stage } > M$}
        	    		\State{ Get the traversed sequence and its}
			\State{ accumulated cost}
            		\State{$\{ \pi^p, f(\pi^p) \} \gets \text{ ReadStack() }$} \Comment{refer \eqref{eq_bp_4}.}
            		\If { $f(\pi^p) \le C_{opt}$}{
            			\State{$\pi^{opt} \gets \pi^p$}
				\State{$C_{opt} \gets f(\pi^p)$}
            		\EndIf
	   		\State{$\text{pop() and }$}
            		\Return
      		\EndIf
      		\State{$\{ Xc_1, Xc_2, Cc_1, Cc_2\} \gets \text{ findBestChildren}(X_{\text{curr}},c,p)$}
      		\State{$\text{Push($Xc_1,Cc_1$,stage)}$}
      		\State{$\text{TraverseTree($Xc_1,Cc_1$,$p$,$q$,$M$,stage$+1$,rec)}$}
      		\If {$\text{rec } = 1$}
      			\State{$\text{Push($Xc_2,Cc_2$,stage)}$}
      			\State{$\text{TraverseTree($Xc_2,Cc_2$,$p$,$q$,$M$,stage$+1$,0)}$}
      		\EndIf
      		\State{$\text{pop() and }$}
      		\Return  
      \EndFunction
      %
      %
      \Function{findBestChildren}{$X_{\text{curr}}$,$c$,$p$}
      	 \State $X_{\text{curr}}  \gets \text{ The current node in the tree being processed}$
         \State $c  \gets \text{ Running cost of the sequence}$
      	 \State $p \gets \text{ The transition statistics}$
	 \For{ each child $X_{i}$ of the current node $X_{\text{curr}}$}
		\State{$c(i) \gets c(i) +  p(X_{i},X_{\text{curr}})\log_2 \frac{p(X_{i},X_{\text{curr}})}{p(X_{i})p(X_{\text{curr}})}$}
         \EndFor
        \State{$\{ I_1, I_2 \} \gets argmax(c)$}
        \State{ Return two best children and their running cost.}
        \State{\textbf{return} $\{ X_{I_1}, X_{I_2}, c(I_1), c(I_2) \}$} 
      \EndFunction  
  }       
 \end{algorithmic}
\end{algorithm}
\end{savenotes}
\end{minipage}

\section{Computational complexity analysis}\label{CCA}
\begin{figure}[b!]
\centering
\includegraphics[scale=0.25]{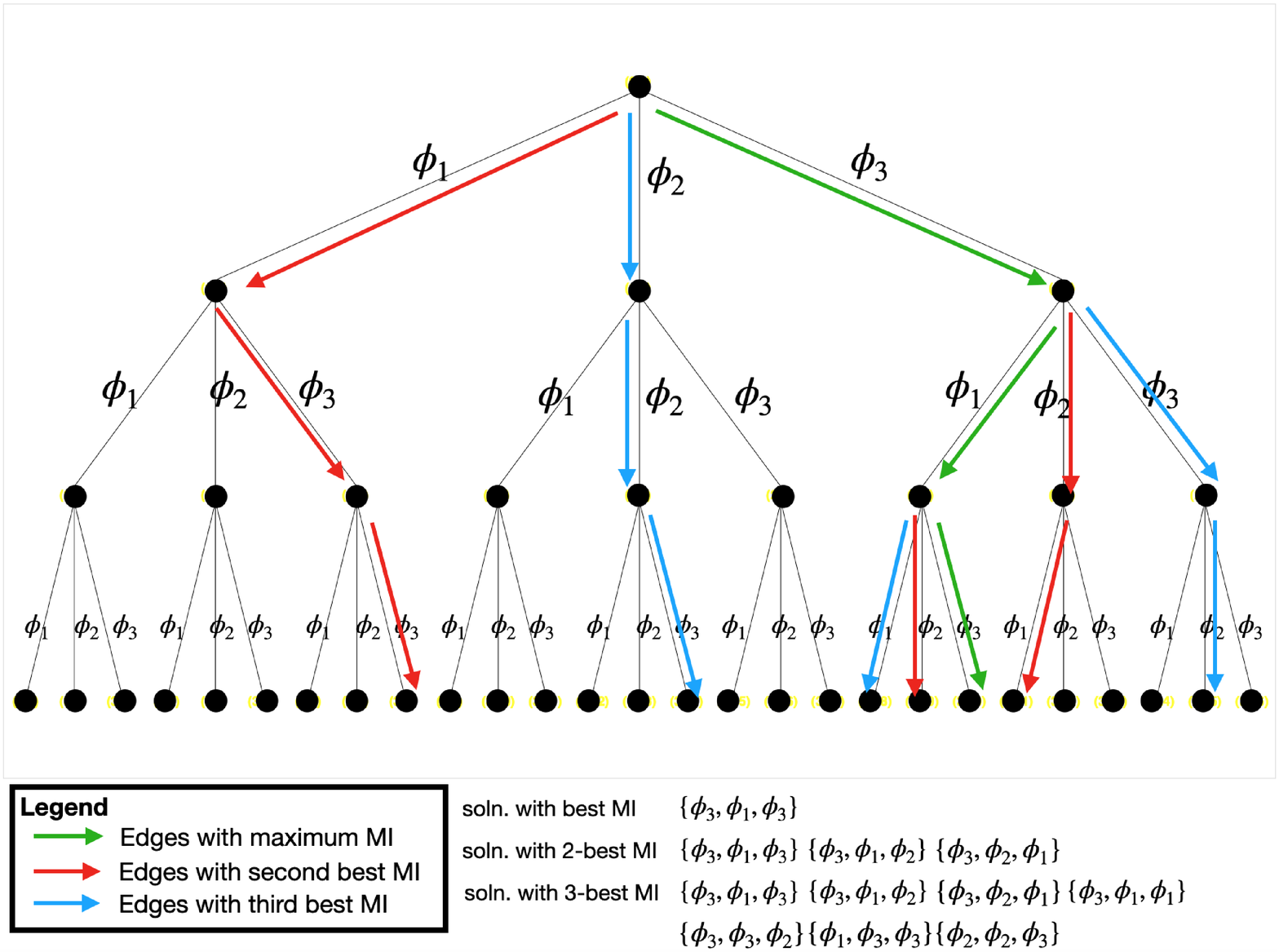}
\caption{An illustration of the path (solutions) explored when using a single-pass, $2-$best, and $3-$best children traversal.}
\label{treecomp}
\end{figure}
The algorithm yields an optimal solution in probability $\pi^{opt}$ if the priors $q$ selected is a close representation of the optimal solution $\pi^*$. In such a situation, the proposed IDBP algorithm requires a single-pass tree traversal to get to the solution $\pi^{opt}$. This is the best case. However, when $q$ is not an accurate representation of $\pi^*$, additional solutions can be explored using a second pass from every node visited by traversing the tree along the second-best child. Although following the path along second-best child recursively explores more solutions, it is easy to see that this increases the complexity exponentially in $M$, having a time complexity of $\approx O(2^M)$. The algorithm will turn out to be an ES if one has to follow $K-$best paths recursively having a complexity of  $O(K^M)$. Alternatively, we propose to follow $k-best$ children, but not recursively. A $2-$best children solution exploration in a non-recursive fashion is described in Algorithm \ref{Algo1}. One can choose to extend this algorithm to explore $k$-best children. This is illustrated using Fig.\ref{treecomp}.\\
\indent A single-pass tree traversal to get to the solution $\pi^{opt}$ has a complexity of $O(\mu KM)$, where $M$ is the number of RIS elements (also the depth of the tree under consideration). The term $\mu$ is the number of arithmetic operations required to compute the MI between the current node and one of its children. Hence to compute the MI between a given node and all its children, the number of arithmetic operations required is $\mu K$, where $K$ is the cardinality of $\Phi$. When exploring additional solutions using a second pass from every node visited (in a non-recursive fashion) to traverse the tree along the second-best child, the number of nodes to be processed is $M + 1 + 2 + \cdots + M-1 = \frac{M(M+1)}{2}$, and hence has a complexity of $O(\mu K M^2)$. This is illustrated in Fig.\ref{treecomp}. Similarly, when we consider solutions from the $3-$best children along the best-child path, the number of nodes to be processed is $M + 2 + 4 + \cdots + 2(M-1) = M + 2\frac{M(M-1)}{2}$, which again has $O(\mu K M^2)$ complexity. In general, solutions considering $k-$best children along the best-child path have a complexity of $O(\mu K k M^2)$. Extending the result to explore $K-$best solutions from the best path still has a polynomial-time computational complexity of $O(\mu K^2M^2)$. On average, with a prior $q$ selected to have a close statistics of the optimal solution $\pi^*$, the proposed IDBP algorithm yields an optimal solution in probability $\pi^{opt}$ with a complexity of $O(\mu K^2 M^2)$.\\
One of the many ways to identify the priors $q$ to have a good representation of $\pi^*$ is to use a fast heuristic algorithm to identify $\{ \pi^i \}_{i=1}^m$ discussed in subsection \ref{qeval} \cite{SimAn2}. It is to be noted that this computational complexity does not include the evaluation of the conditional priors $q$ described in \ref{qeval}. The priors $q$ can be evaluated with significantly reduced computation using random sampling (with $m \ll M$) or heuristics methods \cite{qevalm,qevalNN}.\\
\indent The TMH algorithm proposed in \cite{RisASyed} requires the computation of the matrix $\bold{K}$, and finding its eigenvector that corresponds to its maximum eigenvalue as described using (11) and (12) in Section III-A of \cite{RisASyed}. The resultant eigenvector quantized to the nearest possible discrete angles yields the solution. To compute the matrix $\bold{K}$ the effective number of multiplications are $N_t N_r M^2$. Finding the required eigenvector has a complexity of $O(M^3)$, assuming no structure about the matrix $\bold{K}$, which is a reasonable assumption. This results in the computational complexity of TMH to be $O(M^3)$.\\
\indent The complexity of the reflecting schemes \textit{eMSER} and \textit{vMSER} proposed in \cite{JointRISPrec} is shown to be $\approx O(L^{2N}K^2)$. Here $L$ corresponds to the $L-ary$ QAM symbols used. The discussion is summarized in the Table \ref{compcca}.
\begin{table}[!htb]
\begin{center}
\resizebox{0.95\columnwidth}{!}{%
\begin{tabu} to 0.2\textwidth {|c|c|c|}
 \hline
 \textbf{\small Algorithm}  & \textbf{\small  Computational complexity} & \textbf{\small  Matlab runtime* for $M=12$} \\
 \hline
 \footnotesize  ES & \footnotesize  $O(K^M)$ & \footnotesize  415.5 \\
 \hline
 \footnotesize  Proposed IDBP & \footnotesize  $\approx O(KM)$ $^{\S}$ & \footnotesize 18.3 \\
 \hline
 \footnotesize  Proposed IDBP & \footnotesize  $\approx O(K^2M^2)$ $^{\dagger}$ & \footnotesize 18.8 \\
 \hline
\footnotesize  TMH$^{\lozenge}$ &  \footnotesize  $\approx O(M^3)$ & \footnotesize 56 \\
 \hline
\footnotesize  AO1$^{\lozenge}$ \textit{(eMSER-Reflecting)} &  \footnotesize  $\approx O(L^{2N}K^2) / O(L^{2N}K^3)$ & \footnotesize  228 \\
 \hline 
 \footnotesize  AO2$^{\lozenge}$ \textit{(eMSER-Reflecting)} &  \footnotesize $\approx O(L^{2N}K^2) / O(L^{2N}K^3)$ & \footnotesize 345 \\
 \hline 
\end{tabu}}\\
\end{center}
\footnotesize{$^{\S}$ conditional priors $q$ is a a close representation of the solution $\pi^*$,\\
$^{\dagger}$ conditional priors $q$ not a close representation of the solution $\pi^*$.\\
$^{\lozenge}$ refer to Section \ref{Sim} (Simulations).\\
$*$The matlab runtime (in secs.) includes precoder, combiner, RIS evaluations, and prior evaluation at a given SNR.}
\caption{$\text{\small Computational complexity comparison.}$}\label{compcca}
\end{table}

\section{Simulations}\label{Sim}
\begin{figure}[b!]
\centering
\begin{subfigure}[b]{.5\textwidth}
\includegraphics[scale=0.42]{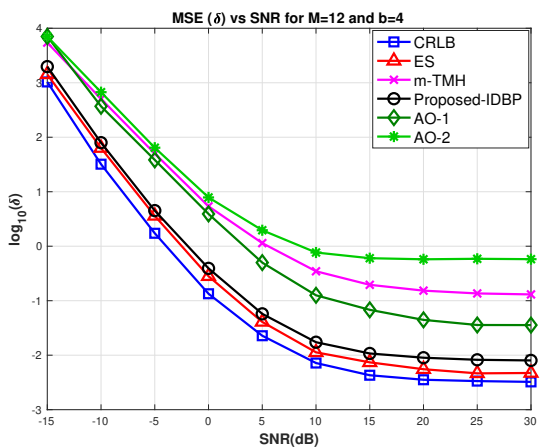}
\subcaption{\scriptsize MSE performance.}\label{fig_m12_ES}
\end{subfigure}
\vfill 
\begin{subfigure}[b]{.5\textwidth}
\includegraphics[scale=0.42]{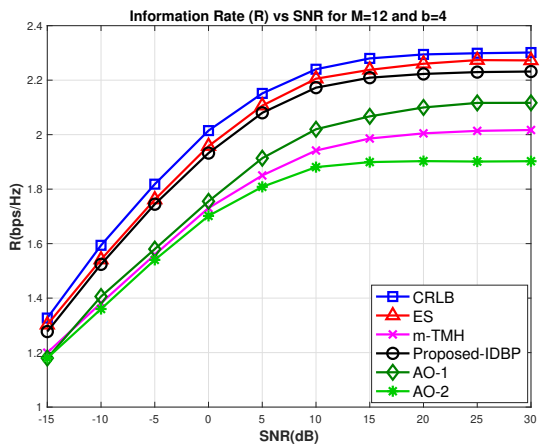}
\subcaption{\scriptsize Information rate performance.}\label{fig_r12_ES}
\end{subfigure}
\caption{\footnotesize MSE and information rate at various SNRs with proposed IDBP,  TMH, AO, and the ES method with the number of RIS elements $M=12$ for ADC bits $b=4$ on all RF paths.}\label{fig_12_ES}
\end{figure}

In this section, we first compare the following algorithms- (i) the exhaustive search (ES) method to solve the \eqref{eq_21}, (ii) the proposed IDBP algorithm to solve \eqref{eq_21} (IDBP), (iii) the exhaustive search to solve the trace maximization (TM) framework considering the diagonal RIS architecture proposed in \cite{RisASyed} (m-TMH), and the AO algorithm proposed in \cite{JointRISPrec}. The evaluation of the ES for RIS elements when $M > 12$, and with the phase-shift settings $K \ge 3$ becomes impractical. Hence, for this evaluation, we only consider the case where $M = 12$ with the ADC bit resolution set to $b=4$ on all the RF paths of the receiver. The other configurations parameters used for this evaluation are presented in Table \ref{RisASyedTable}. The channel model for $\bold{P}$ and $\bold{R}$ are derived using the multi-user interference model discussed in Section  \ref{sigmod} considering eight ($\beta = 8$) strong RIS reflected interference and one non-RIS reflected interferer. The detailed analysis of such a multi path propagation environment is described in Section II-B of \cite{RisASyed}. The AO algorithm encompasses the combiner in addition to precoder and RIS that is discussed in \cite{JointRISPrec}. The algorithm is described in Appendix \ref{AppD}. The convergence of the AO algorithm is strongly dependent on the selection of the initial solutions and hence we consider two scenarios of AO with different initial solutions (AO1 and AO2). The initial solutions for AO1 and AO2 are chosen empirically. We run the simulations considering the above parameters to evaluate the MSE, using which we compute the information rate of the link as
\begin{equation}\label{eq_ir}
\begin{split}
R(\bold{\Phi}) =N\log_2 p + \log_2\det \Big ( (\bold{M}(\bold{x}))^{-1} + \frac{1}{p}\bold{I}_{N} \Big).
\end{split}
\end{equation}
The proof of \eqref{eq_ir} is detailed in Appendix \ref{AppB}. The simulation results obtained are shown in Fig. \ref{fig_12_ES}. From Fig. \ref{fig_12_ES}, it can be observed that the ES achieves the CRLB for the given (designed) hybrid precoders and combiners. The proposed IDBP algorithm, which is a computationally efficient method to solve \eqref{eq_21}, extracts a near-optimal solution that is close to ES and has a superior performance compared to both the trace-maximization algorithm (m-TMH) proposed in \cite{RisASyed}, and AO1 and AO2 based on \cite{JointRISPrec}. 
\indent Subsequently, we run simulations with $M = 64, 128, \text{ and }256$ to compare the following algorithms- (i) the proposed IDBP algorithm to solve \eqref{eq_21} (IDBP), (ii) optimal trace maximization (TMH) method called the diagonal $\bold{\Phi}$ (OPT-DIAG) \cite{RisASyed}, and (iii) alternating optimization (AO1) based on the work in \cite{JointRISPrec}. The TMH is a computationally efficient algorithm to solve the trace maximization proposed in \cite{RisASyed}. The details of this algorithm are presented in the Section III-A of \cite{RisASyed}. We evaluate the MSE and the information rate $R$ for SNRs in the range $[-30, 30]$ dB in steps of 5dB and for ADC bits $b=2,3,\text{ and }4$ on all the RF paths. The results obtained are shown using the Fig.\ref{fig_64}, Fig.\ref{fig_128}, and Fig.\ref{fig_256} for $M=64, 128,\text{ and }256$, respectively. From the results, it can be observed that the proposed IDBP algorithm outperforms both the TMH and the AO methods.\\
\begin{table}[H]
\begin{center}
\resizebox{0.95\columnwidth}{!}{%
\begin{tabu} to 0.5\textwidth {| l| l| }
 \hline
 \textbf{Parameters}  & \textbf{Value/Type} \\
 \hline
Frequency & 28Ghz \\
\hline
Environment & Non Line of sight (NLOS) \\
\hline
Tx-Rx seperation & 100m\\
\hline
Tx-RIS seperation & 70m\\
\hline
RIS-Rx seperation & 70m\\
\hline
TX/RX array type & ULA\\
\hline 
Num of TX/RX elements $N_t$/$N_r$ & 48/48\\
\hline
TX/RX  antenna spacing & $\lambda/2$\\
\hline
Number of Passive RIS elements ($M$) & 12,64,128,256\\
\hline
Number of discrete phase settings ($K$) & $\{ \frac{25\pi}{36}, \frac{73\pi}{36},\frac{49\pi}{36}\}$ \\
\hline
ADC bit resolution on all RF paths ($b$) & 2,3,4 \\
\hline
Number of RF paths at TX and RX ($N$) & 8 \\
\hline
Signal bandwidth & 100Mhz \\
\hline
Sampling Frequency &  400Mhz \\
\hline
Modulation Type &  64 QAM \\
\hline
Number of symbols &  200 \\
\hline
Number of interferer paths ($\beta$)  &  8 \\
\hline
\end{tabu}}
\caption{$\text{\footnotesize The configuration parameters used for our simulations.}$}\label{RisASyedTable}
\end{center}
\end{table}
%
%
\begin{figure*}[t!]
\centering
\begin{minipage}[b]{.43\textwidth}
\includegraphics[scale=0.42]{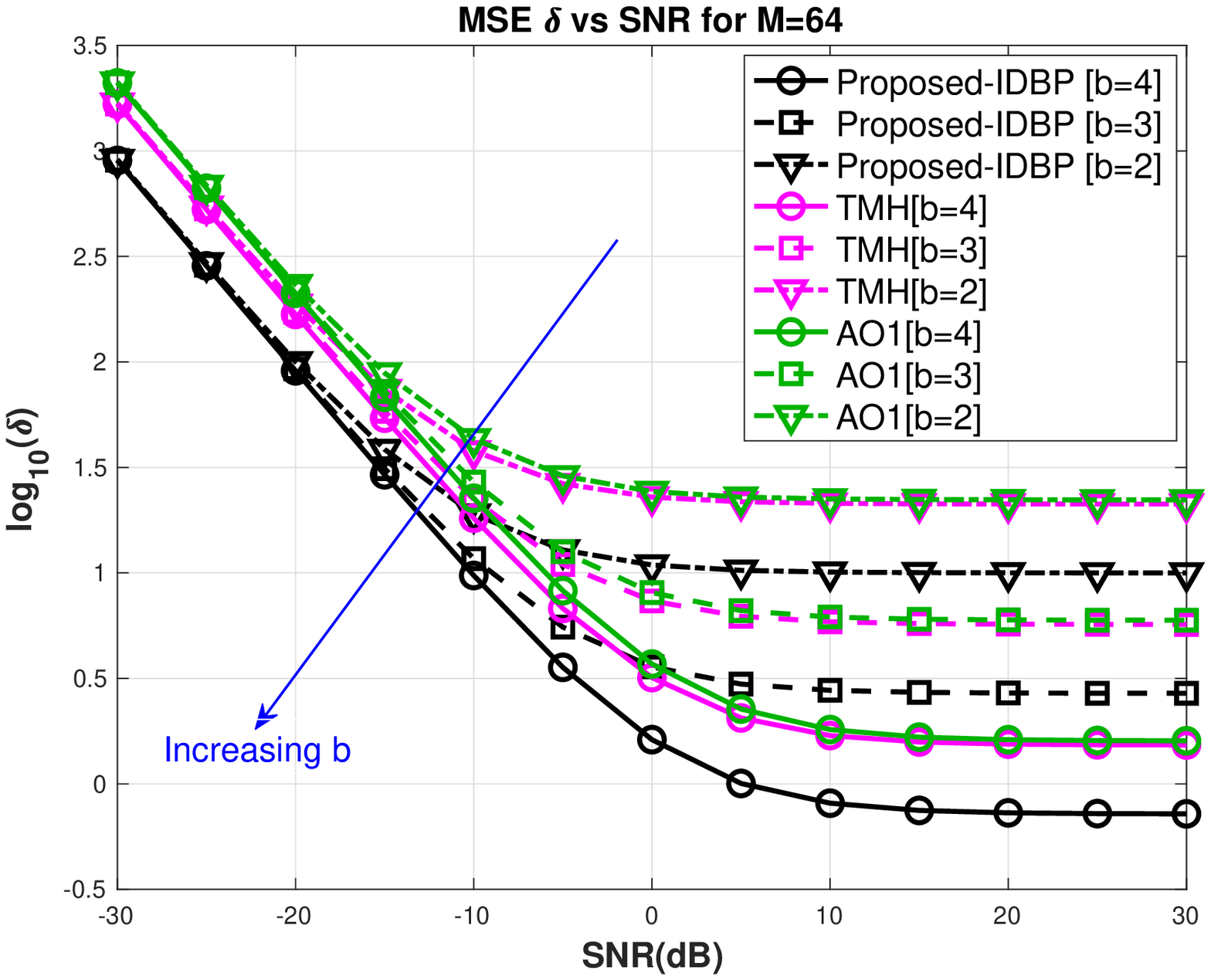}
\subcaption{\scriptsize MSE performance.}\label{fig_m64}
\end{minipage}\qquad
\begin{minipage}[b]{.4\textwidth}
\includegraphics[scale=0.42]{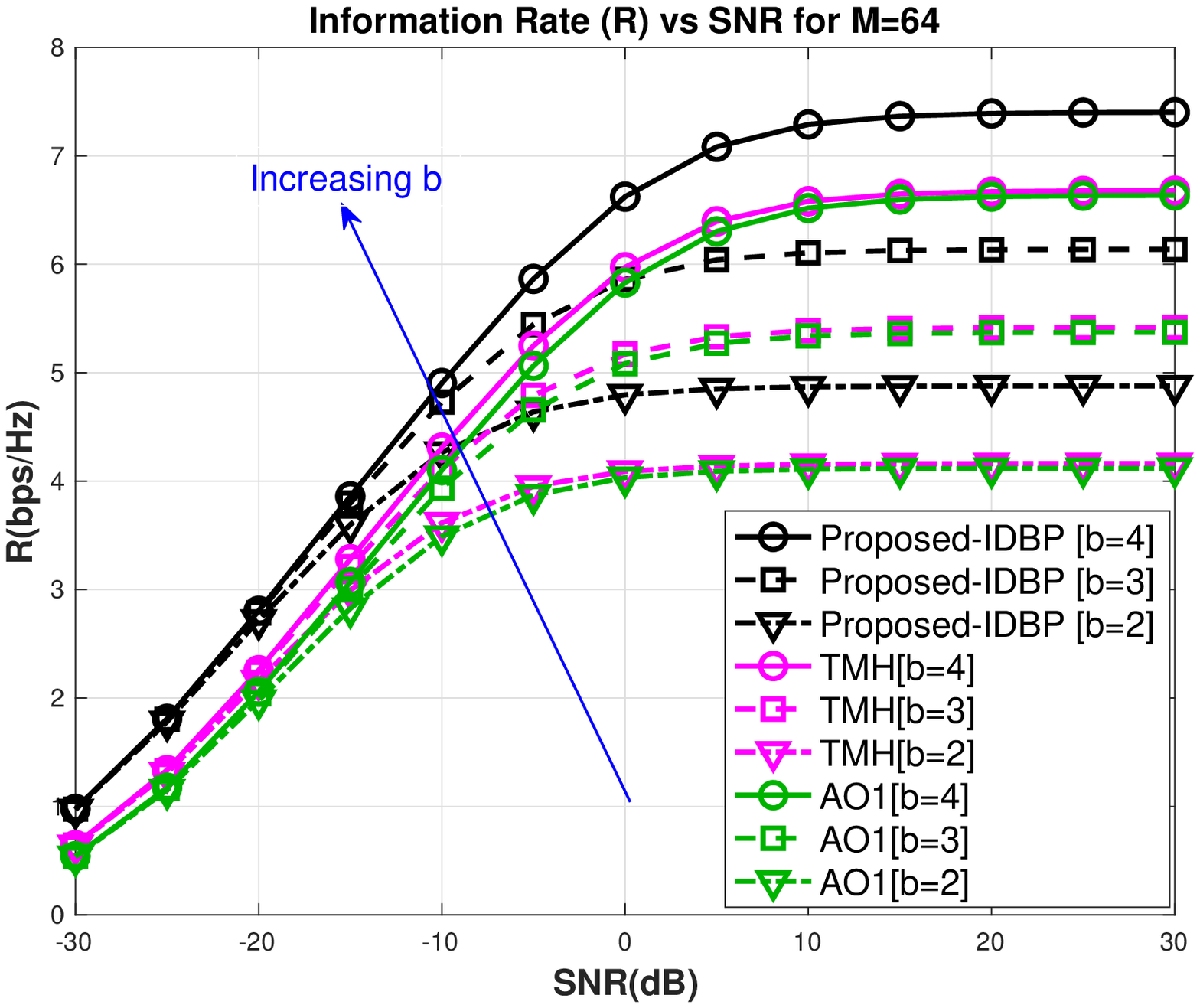}
\subcaption{\scriptsize Information rate performance.}\label{fig_r64}
\end{minipage}
\caption{\footnotesize MSE and information rate at various SNRs with proposed IDBP, TMH, and AO algorithms with the number of RIS elements $M=64$, and for $b-$bit ADC in all of the receiver paths.}\label{fig_64}
\bigskip
\begin{minipage}[b]{.43\textwidth}
\includegraphics[scale=0.42]{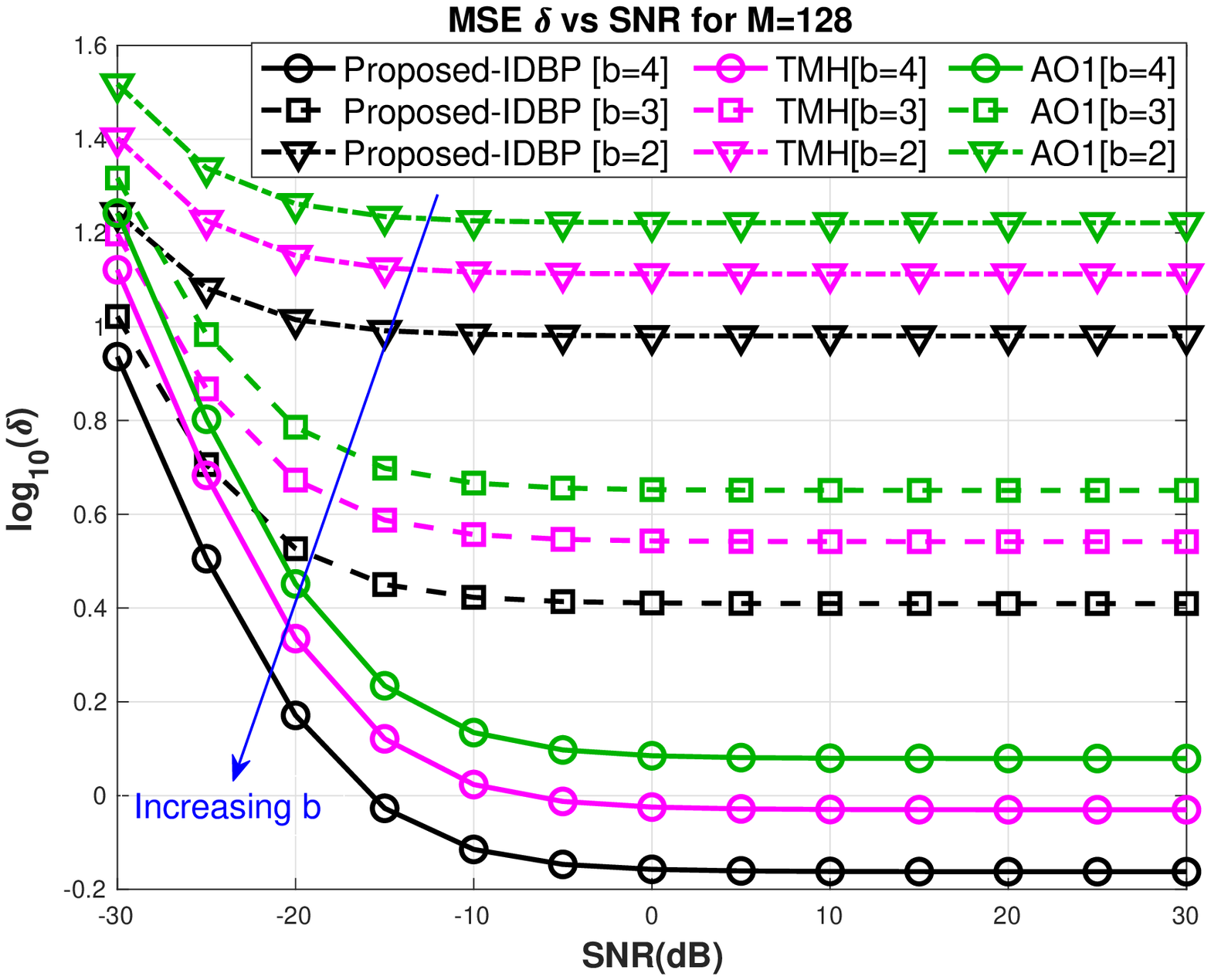}
\subcaption{\scriptsize MSE performance.}\label{fig_m128}
\end{minipage}\qquad
\begin{minipage}[b]{.4\textwidth}
\includegraphics[scale=0.4]{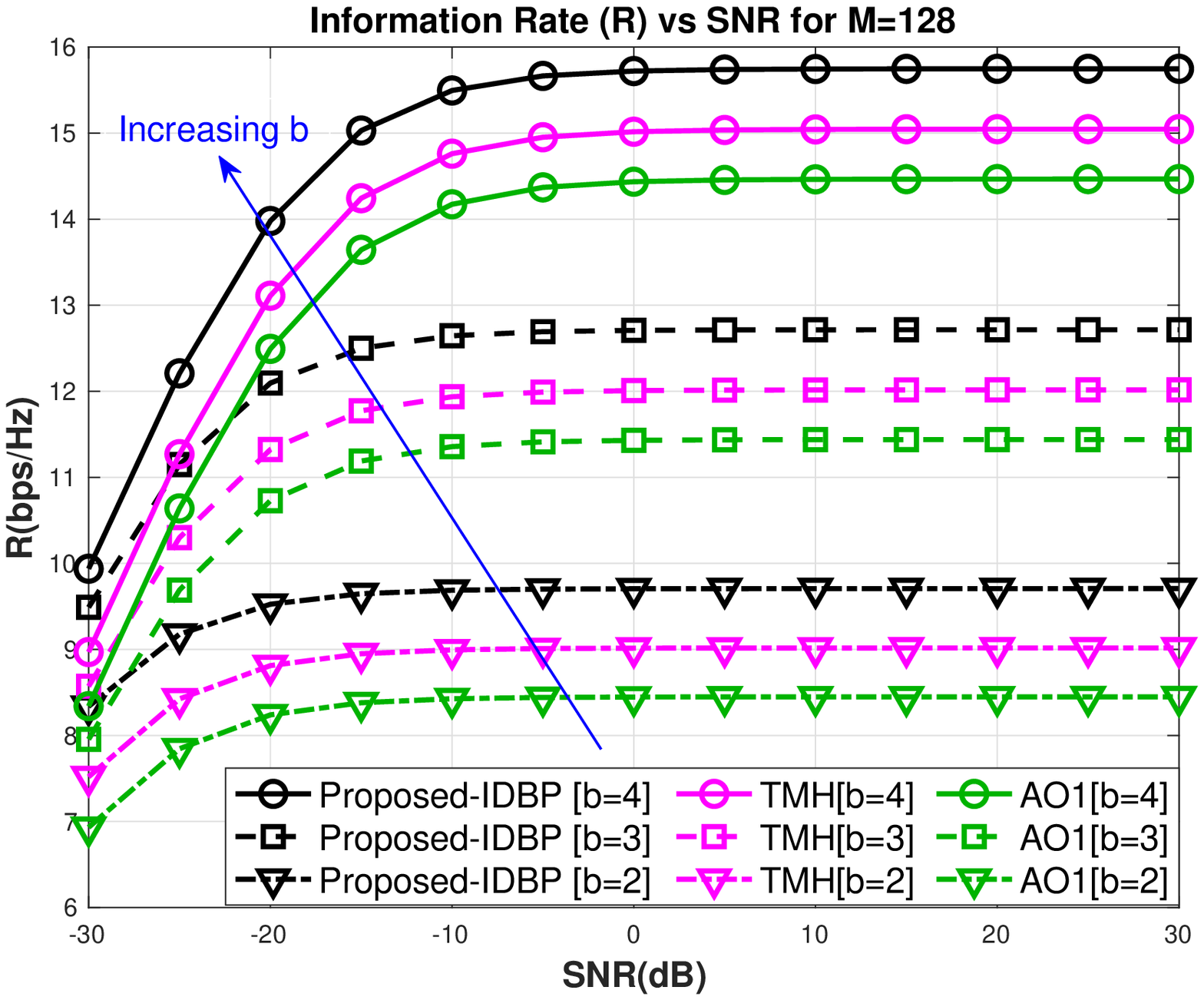}
\subcaption{\scriptsize Information rate performance.}\label{fig_r128}
\end{minipage}
\caption{\footnotesize MSE and information rate at various SNRs with proposed IDBP, TMH, and the AO algorithms with the number of RIS elements $M=128$, and for $b-$bit ADC in all of the receiver paths.}\label{fig_128}
\bigskip
\begin{minipage}[b]{.43\textwidth}
\includegraphics[scale=0.42]{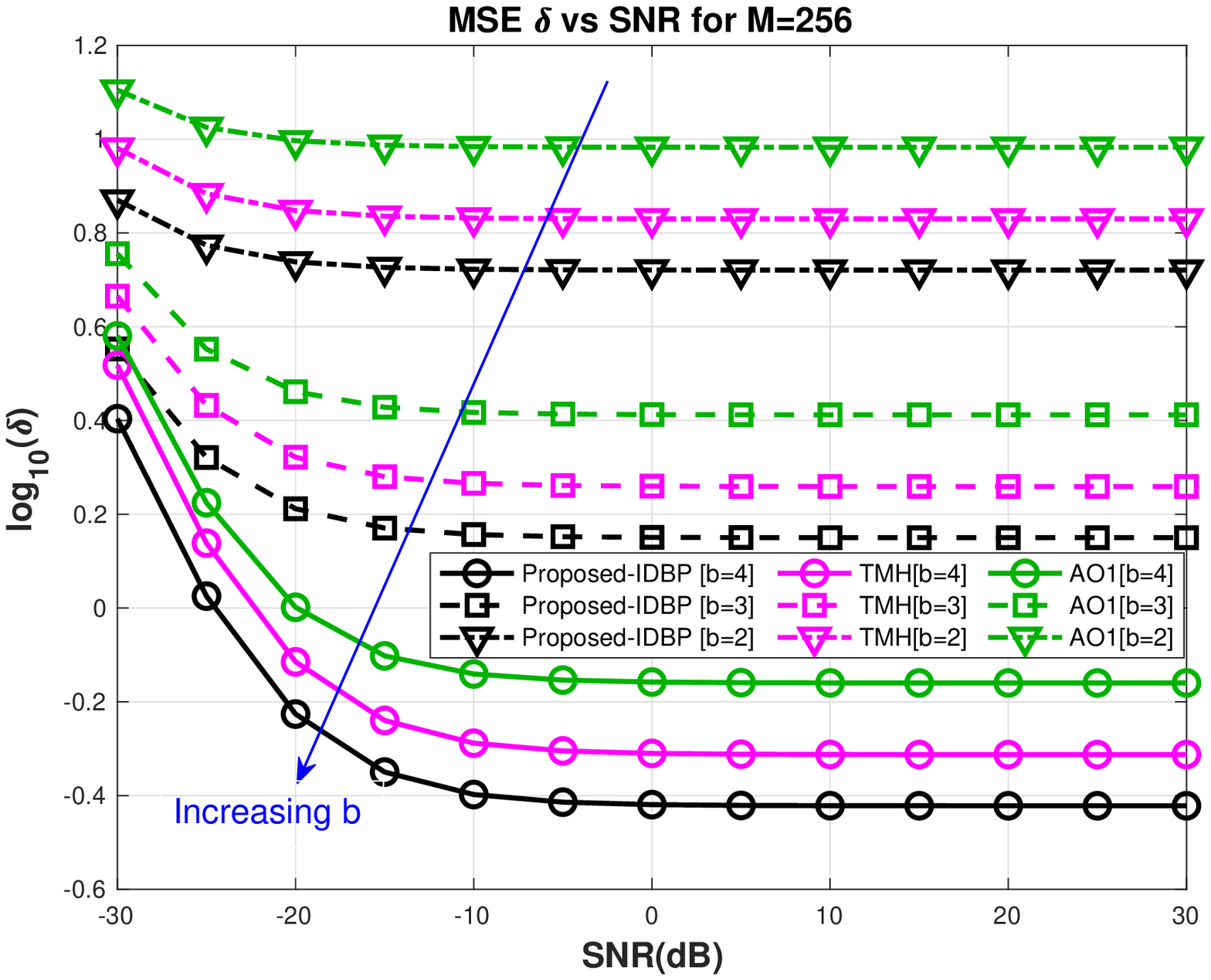}
\subcaption{\scriptsize MSE performance.}\label{fig_m256}
\end{minipage}\qquad
\begin{minipage}[b]{.4\textwidth}
\includegraphics[scale=0.4]{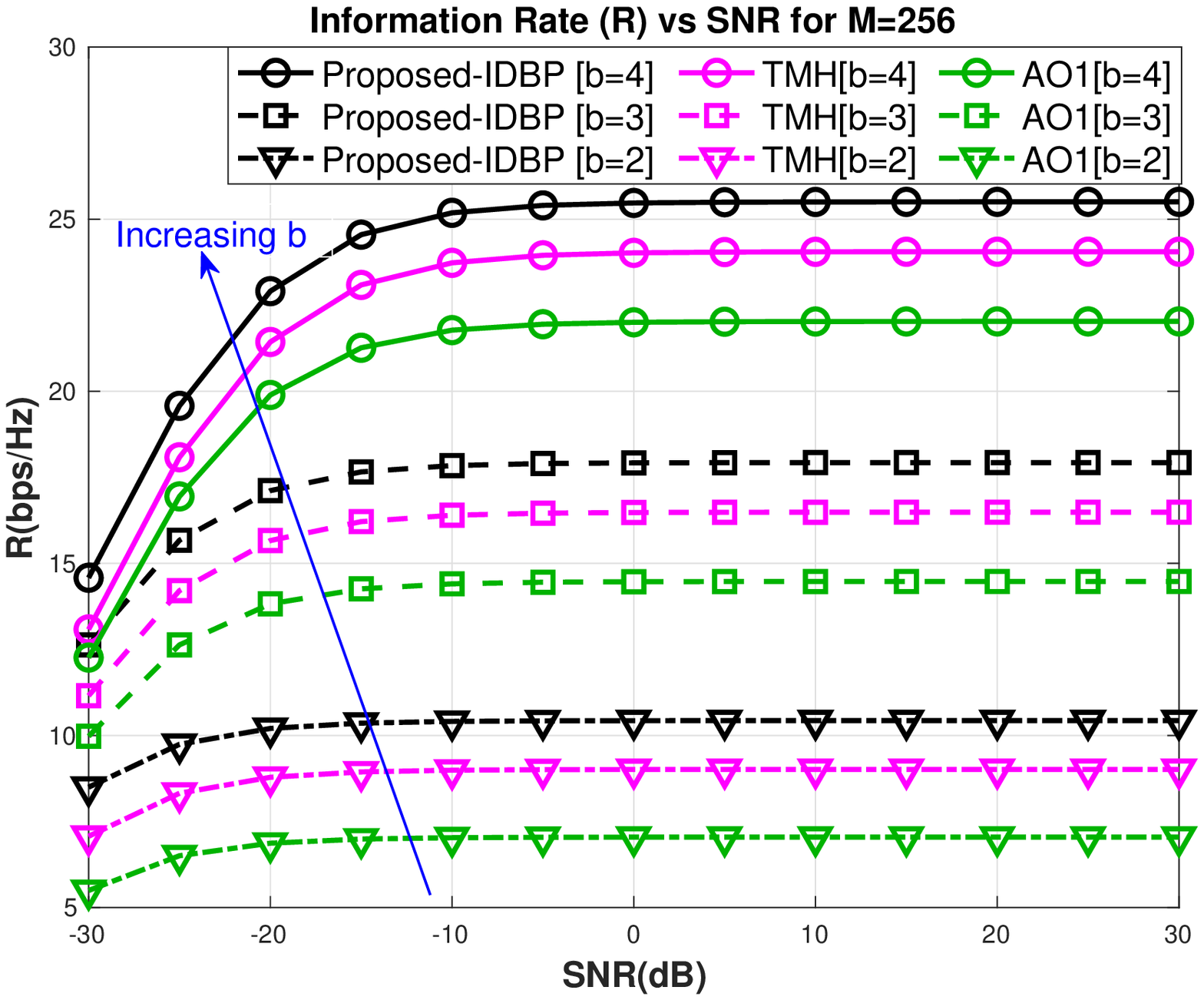}
\subcaption{\scriptsize Information rate performance.}\label{fig_r256}
\end{minipage}
\caption{\footnotesize MSE and information rate at various SNRs with proposed IDBP, TMH, and the AO algorithms with the number of RIS elements $M=256$, and for $b-$bit ADC in all of the receiver paths.}\label{fig_256}
\end{figure*}
%
%

\section{Conclusion}\label{Conc}
The discrete phase optimization algorithm for a passive-RIS that assists a multi-user MaMIMO communication system under interference is studied in this paper. A novel algorithm based on an information-theoretic tree search called IDBP is proposed in this work. We developed the theoretical framework for this proposed algorithm using the Asymptotic Equipartition Property, and establish near-optimality guarantees. We discuss a method to design hybrid precoders and combiners along with RIS phase configuration to minimize the MSE of a blocked LOS link assisted by a RIS and show that it achieves the Cramer-Rao Lower-Bound. We consider the MaMIMO receivers to be equipped with low-resolution ADCs. We also show minimizing the MSE and maximizing the throughput of a blocked LOS link under interference are equivalent. Using simulation, we compare the proposed algorithm with the exhaustive search and two other state-of-the-art algorithms and demonstrate that the proposed method outperforms the state-of-the-art with significant computational savings with an appropriate selection of the prior distribution. This makes it more suitable for the proposed algorithm to be used with RIS having a large number of elements $M$ and a large number of configurable discrete phase settings $K$. The typical use cases of such scenarios are in vehicular and cellular backhaul wireless communication links.


%
\bibliographystyle{IEEEtran}
\bibliography{ITVT_BibTexFile}
\appendix
\subsection{Expression for CRLB}\label{AppA}
We have $\bold{K} = \alpha\bold{W}_S\bold{\bold{\Phi}}\bold{F}_S, \bold{K}^{-1} = \frac{1}{\alpha}\bold{F}_S^{-1}\bold{\bold{\Phi}^{-1}}\bold{W}_S^{-1},$
\begin{equation}\label{eq_app1}
\begin{split}
&(\bold{K}^H)^{-1} = \frac{1}{\alpha}(\bold{W}_S^H)^{-1}\bold{\Phi}(\bold{F}_S^H)^{-1},\\
& \bold{C} = \alpha^2\sigma_n^2\bold{W}\bold{W}^H +  \bold{W}_S\bold{\tilde{W}}_D^H\bold{D}_q^2\bold{\tilde{W}}_D\bold{W}_S^H,\\
\end{split}
\end{equation}
Substituting the terms in \eqref{eq_app1} for the CRLB expression, we have  
\begin{equation}\label{eq_app2}
\begin{split}
&{\bold{I}^{-1}({\bold{\hat{x}}})} = ({\bold{K}^H}{\bold{C}^{-1}}{\bold{K}})^{-1} = \bold{K}^{-1}\bold{C}(\bold{K}^H)^{-1},\\
&= \frac{1}{\alpha}\bold{F}_S^{-1}\bold{\bold{\Phi}^{-1}}\bold{W}_S^{-1}\Big[ \alpha^2\sigma_n^2\bold{W}\bold{W}^H \Big]\frac{1}{\alpha}(\bold{W}_S^H)^{-1}\bold{\Phi}(\bold{F}_S^H)^{-1}+\\
&\frac{1}{\alpha^2}\bold{F}_S^{-1}\bold{\bold{\Phi}^{-1}}\bold{W}_S^{-1}\Big[ \bold{W}_S\bold{\tilde{W}}_D^H\bold{D}_q^2\bold{\tilde{W}}_D\bold{W}_S^H \Big](\bold{W}_S^H)^{-1}\bold{\Phi}(\bold{F}_S^H)^{-1},\\
&= \bold{F}_S^{-1}\Big[ \sigma_n^2\bold{\bold{\Phi}^{-1}}\bold{\tilde{W}}\bold{\Phi}  + \frac{1}{\alpha^2} \bold{\Phi}^{-1}\bold{\tilde{W}}_D^H\bold{D}_q^2\bold{\tilde{W}}_D\bold{\Phi} \Big](\bold{F}_S^H)^{-1},\\
&\text{ where }\bold{\tilde{W}} = \bold{\tilde{W}}_D^H\bold{W}_A^H\bold{W}_A\bold{\tilde{W}}_D.
\end{split}
\end{equation}

\subsection{Expression for the information rate and energy efficiency}\label{AppB}
Considering \eqref{eq_13}, we can write the expression for the information-rate of the MaMIMO channel as function of the RIS phase shift matrix $\bold{\Phi}$ as \cite{Zakir7}
\begin{equation}\label{eq_app3}
\begin{split}
R(\bold{\Phi}) &= I\big(\bold{x}; \bold{y}\big) = h(\bold{y}) - h(\bold{y}|\bold{x})\\
&= h(\bold{y}) - h(\bold{K}\bold{x} + \bold{n}_1|\bold{x}) \overset{(a)} = h(\bold{y}) - h(\bold{n}_1),
\end{split}
\end{equation}
where $I\big(\bold{x}; \bold{y}\big)$ is the mutual information of random variables $\bold{x}$ and $\bold{y}$, and $\bold{K}$ is a function of the RIS phase shift matrix $\bold{\Phi}$. (a) holds if and only if both $\bold{n}_q$ and $\bold{x}$ are Gaussian. Hence, ensures $\bold{y}$ is Gaussian. Also, if $\bold{y} \in \mathbb{C}^{N}$, then the differential entropy $h(\bold{y})$ is less than or equal to $\log_2\det(\pi e \bold{B})$ with equality if and only if $\bold{y}$ is circularly symmetric complex Gaussian with $E[\bold{y}\bold{y}^H] = \bold{B}$ \cite{Bengt}. That is
\begin{equation}\label{eq_app4}
\begin{split}
\bold{B} &= E \Big[ (\bold{K}\bold{x} + \bold{n}_1)(\bold{K}\bold{x} + \bold{n}_1)^H \Big]\\
&= E \Big[ \bold{K}\bold{x}\bold{x}^H\bold{K}^H + \bold{n}_1\bold{n}_1^H \Big] = p\bold{K}\bold{K}^H + \bold{C}.
\end{split}
\end{equation}
where $\bold{C} = \alpha^2\sigma_n^2\bold{W}\bold{W}^H +  \bold{W}_S\bold{\tilde{W}}_D^H\bold{D}_q^2\bold{\tilde{W}}_D\bold{W}_S^H$.
The differential entropies $h(\bold{y})$ and $h(\bold{n}_1)$ satisfy
\begin{equation}\label{eq_app5}
\begin{split}
h(\bold{y}) &\le \log_2\det(\pi e \bold{B}) = \log_2\det \bigg( \pi e \Big(p \bold{K}\bold{K}^H + \bold{C} \Big) \bigg),\\
h(\bold{n}_1) &\le \log_2\det(\pi e \bold{C}),
\end{split}
\end{equation}
with equality iff $\bold{y}$ and $\bold{n}_1$ posses circularly symmetric complex Gaussian statistics. However, using the Theorem-1 in \cite{Zakir7}, it is straightforward to see that $\bold{n}_1 \sim \mathcal{CN}(\bold{0},\bold{C})$. Hence we have
\begin{equation}\label{eq_app6}
h(\bold{n}_1) = \log_2\det(\pi e \bold{C}).
\end{equation}
Thus the expression for the information rate in \eqref{eq_app3} can be rewritten as
\begin{equation}\label{eq_app7}
\begin{split}
R(\bold{\Phi}) =  h(\bold{y}) - h(\bold{n}_1) & \overset{(b)} = \log_2\det( \pi e \bold{B}) - \log_2\det(\pi e \bold{C})\\
 &= \log_2\det \Big ( p\bold{K}\bold{K}^H\bold{C}^{-1} + \bold{I}_{N} \Big),
\end{split}
\end{equation}
where (b) follows from the assumption that the input symbol vector $\bold{x}$ is circular symmetric Gaussian vector that could be modeled 
as $\bold{x} \sim \mathcal{CN}(\bold{0},p\bold{I}_{N})$ \cite{Rangan,VarBitAllocJour,Zakir7}. The information rate in \eqref{eq_app7} can be further simplified as \cite{Zakir7}
\begin{equation}\label{eq_app8} 
\begin{split}
R(\bold{\Phi}) &= \log_2\det \Big ( p\bold{K}\bold{K}^H\bold{C}^{-1}\bold{K}\bold{K}^{-1} + \bold{K}\bold{K}^{-1} \Big),\\
&= \log_2\det \Big ( p\bold{K} \big ( \bold{K}^H\bold{C}^{-1}\bold{K} + \frac{1}{p}\bold{I}_{N} \big) \bold{K}^{-1} \Big),\\
&= \log_2\det ( p\bold{K}) \det \Big ( \bold{K}^H\bold{C}^{-1}\bold{K} + \frac{1}{p}\bold{I}_{N} \Big) \det (\bold{K}^{-1}),\\
&= \log_2 p^{N}\det \Big ( \bold{K}^H\bold{C}^{-1}\bold{K} + \frac{1}{p}\bold{I}_{N} \Big),\\
&=N\log_2 p + \log_2\det \Big ( (\bold{I}^{-1}({\bold{\hat{x}}}))^{-1} + \frac{1}{p}\bold{I}_{N} \Big).
\end{split}
\end{equation}
Since the MSE $\bold{M}(\bold{x})$ achieves the CRLB by the design of the precoders and combiners as seen in \eqref{eq_6a}, we can also write the information-rate as follows
\begin{equation}\label{eq_app8}
\begin{split}
R(\bold{\Phi}) = N\log_2 p + \log_2\det \Big ( (\bold{M}(\bold{x}))^{-1} + \frac{1}{p}\bold{I}_{N} \Big).
\end{split}
\end{equation}
Similarly, we can define the energy efficiency (EE) as a function of the RIS phase matrix $\bold{\Phi}$ as
\begin{equation}\label{eq_app9}
\begin{split}
\eta_{EE}(\bold{\Phi}) &= \frac{R(\bold{\Phi})}{p(b)}\text{ (bits/Hz/Joule)}\\
&= \frac{{N}\log_2 p + \log_2\det \Big ( (\bold{M}(\bold{x}))^{-1} + \frac{1}{p}\bold{I}_{N} \Big)}{P_T + P_R + P_{RIS} + 2Nc{f_s}2^{b}}, 
\end{split}
\end{equation}
where where $p(b)$ is the total power consumed. Here $P_T$, $P_R$, and $P_{RIS}$ are the power consumed at the transmitter, receiver, and RIS respectively. The net ADC power consumption is $2Nc{f_s}2^{b}$, where $b$ is the ADC bit resolution used in all the $N$ RF paths, $c$ is the power consumed per conversion step and $f_s$ is the sampling rate in Hz \cite{Uplink}.\\
\indent From \eqref{eq_app8} and \eqref{eq_app9}, it can be shown that maximizing the information rate (throughput) $R$ or maximizing the energy efficiency $\eta_{EE}$ for a given (designed) hybrid precoders and combiners is equivalent to minimizing the CRLB ${\bold{I}^{-1}({\bold{\hat{x}}})}$. This is shown using the Lemma \ref{lemm1} below
\begin{Lemma}\label{lemm1}
\begin{equation}\label{eq_app10}
\underbrace{\text{max}}_{\bold{\Phi}}{R(\bold{\Phi})} \Leftrightarrow \underbrace{\text{max}}_{\bold{\Phi}}{\eta_{EE}(\bold{\Phi})} \Leftrightarrow \underbrace{\text{min}}_{\bold{\Phi}}{{\bold{I}^{-1}({\bold{\hat{x}}})}}.
\end{equation}
\end{Lemma}
\begin{proof}
We can decompose the squared MSE matrix $\bold{M}(\bold{x})$ in \eqref{eq_5} as $\bold{M}(\bold{x}) = \bold{B}\bold{\Lambda}\bold{B}^{-1}$, where $\bold{\Lambda} = \diag (\lambda_1, \lambda_2, \cdots, \lambda_{N})$; such that $\{ \lambda_i \}_{i=1}^{N}$  are the eigenvalues of $\bold{M}(\bold{x})$. It is easy to note that $\bold{M}(\bold{x})$ is always a positive semidefinite matrix, and hence the eigenvalues $\{ \lambda_i \}_{i=1}^{N_s}$ are real and positive \cite{LAStrang}. 
We can further write \eqref{eq_5} as
\begin{equation}\label{eq_appC2}
\begin{split}
\delta(\bold{\Phi}) &\triangleq \tr(\bold{M}(\bold{x})) = \tr(\bold{\Lambda}).
\end{split}
\end{equation}
The MSE $\delta$ can be minimized when $\delta_{min}(\bold{\Phi}) = \min_{\bold{\Phi}}{\sum_{i=1}^{N_s} \lambda_i}$. Hence the condition for \eqref{eq_appC2} to be minimized is $\lambda_i \rightarrow 0, \forall i \in [1,N_s]$.\\
\indent Now, to maximize $R(\bold{\Phi})$ in \eqref{eq_app8} we can write
\begin{equation}\label{eq_appC4}
\begin{split}
R_{max}(\bold{\Phi}) = N\log_2 p + \max_{\bold{\Phi}}{\log_2\det \Big ( (\bold{M}(\bold{x}))^{-1} + \frac{1}{p}\bold{I}_{N} \Big)},
\end{split}
\end{equation}
Since the term $N\log_2 p$ is not dependent on $\bold{\Phi}$, and we know that the function $\log_2(\cdot)$ is monotonically increasing, it suffices to maximize the expression \eqref{eq_appC4} to attain $R_{max}(\bold{\Phi})$
\begin{equation}\label{eq_appC4}
\begin{split}
&\bold{\Phi}^{R_{max}} = \argmaxF_{\bold{\Phi}}{ \Big\{ \det \Big ( (\bold{M}(\bold{x}))^{-1} + \frac{1}{p}\bold{I}_{N} \Big) \Big\}}.
\end{split}
\end{equation}
We can write 
\begin{equation}\label{eq_appC5}
\begin{split}
&\det ( \bold{B}\bold{\Lambda}^{-1}\bold{B}^{-1} + \frac{1}{p}\bold{I}_{N} ) = \det ( \bold{B} [ \bold{\Lambda}^{-1} + \frac{1}{p}\bold{I}_{N} ] \bold{B}^{-1}  ),\\
&= \det ( \bold{\Lambda}^{-1} + \frac{1}{p}\bold{I}_{N}) = \prod_{i=1}^{N} \Big( \frac{1}{\lambda_i} + \frac{1}{p} \Big) = \prod_{i=1}^{N} \Big( \frac{p + \lambda_i}{\lambda_i} \Big).
\end{split}
\end{equation}
This implies $\bold{\Phi}^{R_{max}} = \argmaxF_{\bold{\Phi}}{ \Big\{ \prod_{i=1}^{N} \Big( \frac{p + \lambda_i}{\lambda_i} \Big) \Big\}}$. Since the eigenvalues are real and positive, the maximization \eqref{eq_appC4} is achieved for a given $p$, when $\prod_{i=1}^N \lambda_i \rightarrow 0$ or $\lambda_i \rightarrow 0, \forall i \in [1,N]$, which is similar to the condition that was required to minimize the MSE $\delta$.\\
\indent For energy efficiency $\eta_{EE}(\bold{\Phi})$, maximizing the numerator $R(\bold{\Phi})$ is sufficient condition to maximize the same because the denominator does not depend on the $\bold{\Phi}$ and can be treated as constant.
\end{proof}
%
\subsection{The RIS phase optimization as an MDP}\label{AppC}
The problem \eqref{eq_21} can be visualized as a stochastic-sequential-decision-making (SSDM) problem \cite{SDMP}. The solution $\bold{\Phi}$ at a given time or for a channel realization can be thought of as sequence of decisions to be taken to decide the phase-shifts of the $M$ reflecting elements considering a probabilistic model. The phase-shift value of the first element is selected based on the initial probabilities of the phase-shifts. The subsequent elements’ phase-shifts are arrived based on the previous elements’ phase-shift using the prior and conditional distributions. Here the solution $\bold{\Phi}$ can be thought of as a sequence of random variables $\mathit{\Phi} = \{ \Phi_1, \Phi_2, \cdots, \Phi_M \}$, where the discrete random variable $\Phi_i$ has a probability mass function (PMF) $p(\Phi_i)$.  Also, $p(\Phi_i|\Phi_j)$ represents the transition probabilities across the two reflection elements $i$ and $j$. Let the distribution $q(\Phi_1,\cdots,\Phi_{M})$ denote a prior distribution of the optimal solution to \eqref{eq_21}. An estimate of $q(\Phi_1,\cdots,\Phi_{M})$ can be sampled from the solution space of \eqref{eq_21} as described in Section \ref{qeval}. The sequence in which the phase-shifts are decided is shown as
\begin{equation}\label{eq_c1}
\Phi_1 \longrightarrow \Phi_2 \longrightarrow \Phi_3 \longrightarrow \cdots \longrightarrow \Phi_M.
\end{equation}
A measure of information called Information-to-go $(\mathcal{I}_g)$ was introduced in \cite{Tishby_RL}. The term $\mathcal{I}_g$ is associated with a sequence that specifies cumulated information processing cost or bandwidth required to quantify the future decisions and actions. The measure $(\mathcal{I}_g)$ defines how many bits on average the system needs to specify the future states in an SSDP (or its informational regret) with respect to the prior. This is written as
\begin{equation}\label{eq_c2}
\begin{split}
\mathcal{I}^{\bold{\Phi}^m}_g = \mathbb{E}_{p(\Phi_{m+1}, \cdots, \Phi_M | \bold{\Phi}^m)}\log\frac{p(\Phi_{m+1}, \cdots, \Phi_M | \bold{\Phi}^m)}{q(\Phi_{m+1}, \cdots, \Phi_N)},
\end{split}
\end{equation}
where $p(\Phi_{m+1}, \Phi_{m+2}, \cdots, \Phi_M | \bold{\Phi}^m)$ is the conditional distribution of the future looking sequence given a sequence $\bold{\Phi}^m$, and the fixed prior $q(\Phi_{m+1}, \Phi_{m+2}, \cdots, \Phi_N)$. The term $\bold{\Phi}^m$ indicates the partially observed (decided) sequence $\{ \Phi_1, \Phi_2, \cdots, \Phi_m \}$ for some $m \le M$.\\
However, the analysis with \eqref{eq_c2} is more complex and difficult, hence an approximation to Markovicity is considered \cite{Tishby_RL}. In which case, we can rewrite \eqref{eq_c2} as
\begin{equation}\label{eq_c3}
\begin{split}
&\mathcal{I}^{\bold{\Phi}^m}_g = \mathbb{E}_{p(\Phi_{m+1}, \cdots, \Phi_M | \Phi_m)}\log\frac{p(\Phi_{m+1}, \cdots, \Phi_M | \Phi_m)}{q(\Phi_{m+1}, \cdots, \Phi_M)}.
\end{split}
\end{equation}
In \cite{Tishby_RL}, the authors claim that \textit{"...the Markovicity condition seems, at first sight, a comparatively strong assumption which might seem to limit the applicability of the formalism for modeling the subjective knowledge of an agent. However, under the knowledge of the full state, in the model the agent itself is not assumed to have full access to the state."} (Section 8.2 in \cite{Tishby_RL}).\\
In the case when the prior $q(\Phi_{m+1}, \cdots, \Phi_M)$ can also be sampled as conditionals, that is $q(\Phi_{m+1}, \cdots, \Phi_M | \Phi_m)$, then we can rewrite \eqref{eq_c3} as
\begin{equation}\label{eq_c4}
\begin{split}
\mathcal{I}^{\bold{\Phi}^m}_g = \mathbb{E}_{p(\Phi_{m+1}, \cdots, \Phi_M | \Phi_m)}\log\frac{p(\Phi_{m+1}, \cdots, \Phi_M | \Phi_m)}{q(\Phi_{m+1}, \cdots, \Phi_M | \Phi_m)}.
\end{split}
\end{equation}
Using chain rule and Markovicity, we can establish a recursive relationship for \eqref{eq_c4} \cite{Zakir9}
\begin{equation}\label{eq_c5}
\begin{split}
&\mathcal{I}^{\bold{\Phi}^m}_g = \mathbb{E}_{p(\Phi_{m+1}, \cdots, \Phi_M | \Phi_m)}\log\frac{p(\Phi_{m+1}, \cdots, \Phi_M | \Phi_m)}{q(\Phi_{m+1}, \cdots, \Phi_M | \Phi_m)},\\
& = \mathbb{E}_{p(\Phi_{m+1}, \cdots, \Phi_M | \Phi_m)}\log\frac{p(\Phi_{m+1}|\Phi_{m}) \cdots p(\Phi_{M}|\Phi_{M-1})}{q(\Phi_{m+1}|\Phi_{m}) \cdots q(\Phi_{M}|\Phi_{M-1})},\\
&= \mathbb{E}_{p(\Phi_{m+1}|\Phi_{m})} \log \Bigg[ \frac{p(\Phi_{m+1}|\Phi_{m})}{q(\Phi_{m+1}|\Phi_{m})} \Bigg] + \mathcal{I}^{\bold{\Phi}^{m+1}}_g.
\end{split}
\end{equation}
Hence $\mathcal{I}^{\bold{\Phi}^m}_g$ can be written as a value function with a recursive relationship that satisfy the Bellman's optimality criterion \cite{Tishby_RL, Zakir9} and is a classical example of a MDP. It is also worth noting that effectively \eqref{eq_c5} can be written as
\begin{equation}\label{eq_c6}
\begin{split}
\mathcal{I}^{\bold{\Phi}}_g = D_{KL}(p(\Phi_1,\cdots,\Phi_{M})||q(\Phi_1,\cdots,\Phi_{M})).
\end{split}
\end{equation}
\indent Intuitively, $\mathcal{I}^{\bold{\Phi}}_g \approx 0$ implies that the least information is required to pursue the path $\bold{\Phi}$ for optimality or near-optimality. On the other hand, a large value of $\mathcal{I}^{\bold{\Phi}}_g$ implies considerable information is required to make the decision (or inability to make a decision) in pursuing the path $\bold{\Phi}$ for optimality.
%
\subsection{Alternating optimization}\label{AppD}
The problem in \eqref{eq_d1} can also be solved by updating just one or a few blocks of optimization variables ($\bold{F}_S,\bold{F}_A,\bold{\tilde{F}}_D,\bold{\Phi},{\bold{W}_A^H},{\bold{\tilde{W}}_D^H},\bold{W}_S$) using alternating optimization \cite{JointRISPrec,bcu1}.
\begin{equation}\label{eq_d1}
\delta = \min_{\substack{\bold{F}_S,\bold{F}_A,\bold{\tilde{F}}_D,\bold{\Phi}\\{\bold{W}_A^H},{\bold{\tilde{W}}_D^H},\bold{W}_S}}{\mathcal{L}(\bold{F}_S,\bold{F}_A,\bold{\tilde{F}}_D,\bold{\Phi},{\bold{W}_A^H},{\bold{\tilde{W}}_D^H},\bold{W}_S)}, 
\end{equation}
where $\mathcal{L}(\bold{F}_S,\bold{F}_A,\bold{\tilde{F}}_D,\bold{\Phi},{\bold{W}_A^H},{\bold{\tilde{W}}_D^H},\bold{W}_S) = \tr(\bold{M}(\bold{x}))$.
The algorithm is described below
\begin{breakablealgorithm}
  \caption{Alternating optimization}\label{Algo_BCD}
  \begin{algorithmic}[1]
   \small
      \Procedure{AO}{$\bold{F}_S^0,\bold{F}_A^0,\bold{\tilde{F}}_D^0,\bold{\Phi}^0,{\bold{W}_A^H}^0,\bold{\tilde{W}}_D^{H^0},\bold{W}_S^0,\epsilon_T$}
      		\State{$k \gets 0$}
		\State {$\delta_{k} \gets \mathcal{L}(\bold{F}_S^{k},\bold{F}_A^{k},\bold{\tilde{F}}_D^{k},\bold{\Phi}^{k},{\bold{W}_A^H}^{k},\bold{\tilde{W}}_D^{H^k},\bold{W}_S^{k})$}
      		\Do
	                 \State{$\text{Solve for }\bold{F}_A,\bold{\tilde{F}}_D\text{ using }\textit{MSER-Precoding in}\text{ \cite{JointRISPrec}}$}
	                 \State {$\{{\bold{F}_A^{k+1}},{\bold{\tilde{F}}_D^{k+1}}\} \gets$}
	                 \State {$\argminF_{\bold{F}_A,\bold{\tilde{F}}_D}{\mathcal{L}(\bold{F}_S^k,\bold{F}_A,\bold{\tilde{F}}_D,\bold{\Phi}^k,{\bold{W}_A^H}^k,\bold{\tilde{W}}_D^{H^k},\bold{W}_S^k)}$}
	                 \State{$\text{Solve for }\bold{W}_A^H, \bold{\tilde{W}}_D^H \text{ using \eqref{eq_12}}$}
	                 \State {$\{{\bold{W}_A^H}^{k+1},\bold{\tilde{W}}_D^{H^{k+1}}\} \gets$}
	                 \State {$\argminF_{\bold{W}_A^H, \bold{\tilde{W}}_D^H}{\mathcal{L}(\bold{F}_S^{k},\bold{F}_A^{k+1},\bold{\tilde{F}}_D^{k+1},\bold{\Phi}^k,{\bold{W}_A^H},{\bold{\tilde{W}}_D^H},\bold{W}_S^k)}$}	
	                 \State{$\text{Solve for }\bold{\Phi}\text{ using }\textit{eMSER-Reflecting in}\text{ \cite{JointRISPrec}}$
	                 \State {$\bold{\Phi}^{k+1} \gets$}
	                 \State {$\argminF_{\bold{\Phi}}{\mathcal{L}(\bold{F}_S^{k},\bold{F}_A^{k+1},\bold{\tilde{F}}_D^{k+1},\bold{\Phi},{\bold{W}_A^H}^{k+1},\bold{\tilde{W}}_D^{H^{k+1}},\bold{W}_S^k)}$}
	                 \State{$\text{Solve for }\bold{F}_S, \bold{W}_S\text{ using \eqref{eq_23} and \eqref{eq_24}}$}
	                 \State {$\{ \bold{F}_S^{k+1}, \bold{W}_S^{k+1} \} \gets \argminF_{\bold{F}_S, \bold{W}_S}{\mathcal{L}(\bold{F}_S,\bold{F}_A^{k+1},\bold{\tilde{F}}_D^{k+1},\bold{\Phi}^{k+1}},$} 
	                 \State {${\bold{W}_A^H}^{k+1},\bold{\tilde{W}}_D^{H^{k+1}},\bold{W}_S)$}
	                 \State {$\delta_{k+1} \gets $}
	                 \State{$\mathcal{L}(\bold{F}_S^{k+1},\bold{F}_A^{k+1},\bold{\tilde{F}}_D^{k+1},\bold{\Phi}^{k+1},{\bold{W}_A^H}^{k+1},\bold{\tilde{W}}_D^{H^{k+1}},\bold{W}_S^{k+1})$} 
	                 \State{$err \gets \delta_{k} - \delta_{k+1}$}  
	                 \State{$k \gets k + 1$}   
		\doWhile{$\{ err > \epsilon_T \}$} 
  \EndProcedure}
\end{algorithmic}
\end{breakablealgorithm}

\vskip -3\baselineskip
\begin{IEEEbiography}[{\includegraphics[width=1in,height=1.25in,clip,keepaspectratio,trim=4 4 4 4]{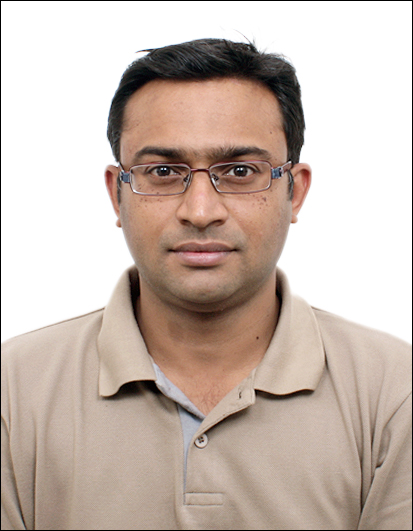}}]%
{I. Zakir Ahmed}
(Member, IEEE) received his Bachelor’s degree in Electrical Engineering from UVCE, Bengaluru, and an M.S. degree from Illinois Institute of Technology, Chicago. He is currently pursuing his Ph.D. Degree in Electrical Engineering from the University of California at Santa Cruz. He has over 20 years of industry experience as a researcher, design engineer, and practitioner in wireless communication and signal processing domains. He worked in the Motorola DSP group in Bangalore, India, from 2000 to 2009. He worked as an RF engineer with National Instruments from 2009 to 2020. Since 2020 he has been based out of San Diego, California, working as a Machine-learning research engineer for applications in RF systems. He holds more than ten US patents and has several research publications in the areas of wireless communication and signal processing. His research interests include optimization methods, signal processing, wireless communication, information theory, and machine learning.
\end{IEEEbiography}
\vskip -3\baselineskip plus -1fil
\begin{IEEEbiography}[{\includegraphics[width=1in,height=1.25in,clip,keepaspectratio,trim=4 4 4 4]{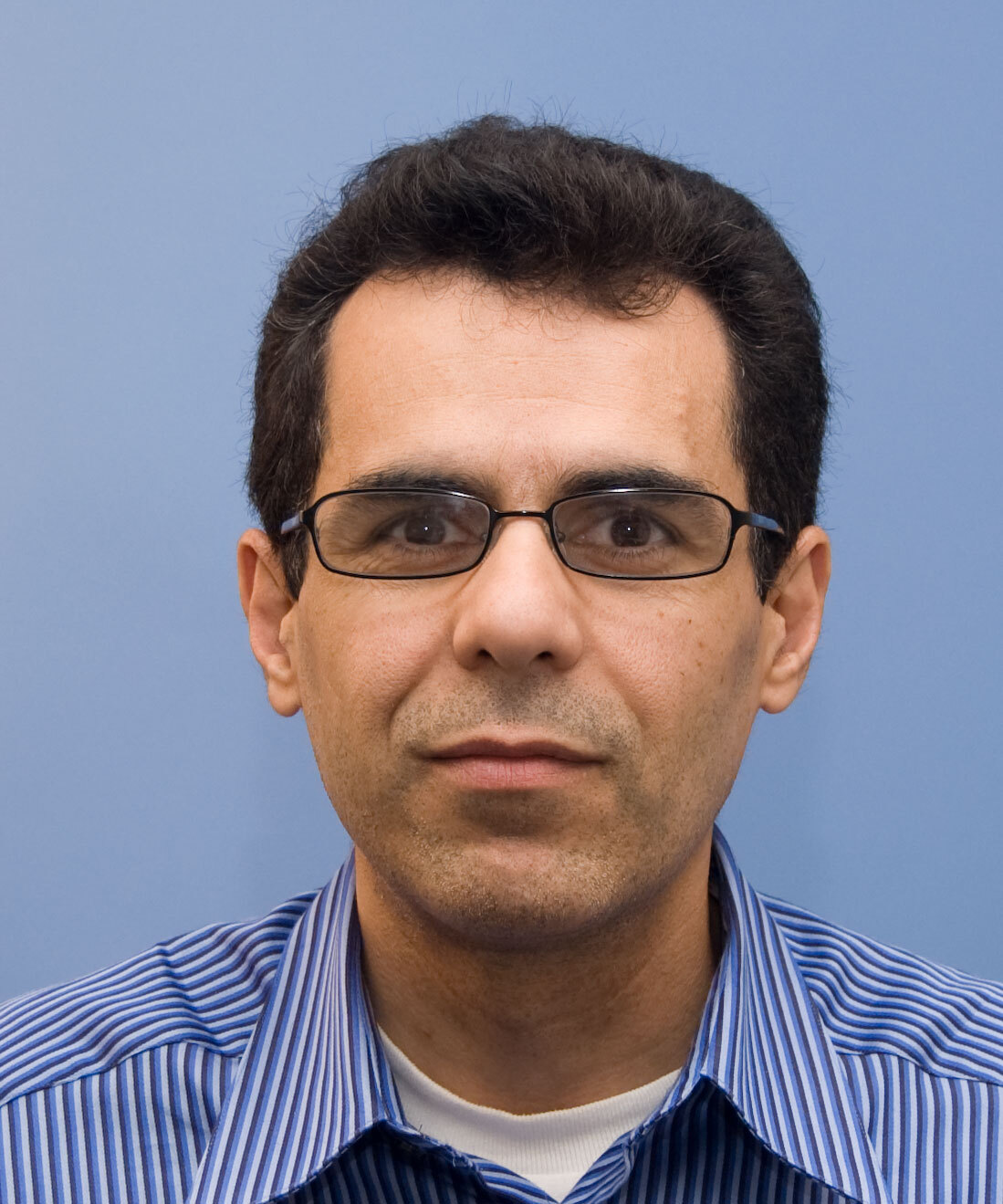}}]%
{Hamid R. Sadjadpour}
(Senior Member, IEEE) received the B.S. and M.S. degrees from Sharif University of Technology, and the Ph.D. degree from University of Southern California (USC). In 1995, he joined the AT\&T Shannon Research Laboratory, as a Technical Staff Member and later as a Principal Member of Technical Staff. In 2001, he joined the University of California at Santa Cruz where he is currently a Professor. He has authored over 200 publications and holds 20 patents. His research interests are in the general areas of wireless communications, security and networks. He has served as a Technical Program Committee Member and the Chair in numerous conferences. He is a co-recipient of the best paper awards at the 2007 International Symposium on Performance Evaluation of Computer and Telecommunication Systems, the 2008 Military Communications conference, the 2010 European Wireless Conference, and the 2017 Conference on Cloud and Big Data Computing.
\end{IEEEbiography}
\vskip -2\baselineskip plus -1fil
\begin{IEEEbiography}[{\includegraphics[width=1in,height=1.25in,clip,keepaspectratio,trim=4 4 4 4]{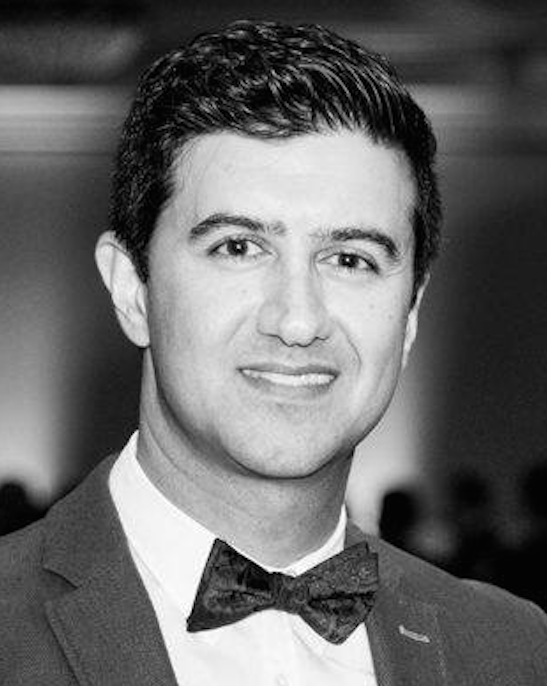}}]%
{Shahram Yousefi}
(Senior Member, IEEE) has 25 years of experience in engineering, technology, teaching, managing teams and projects, as well as public speaking and advising. He received his B.Sc. in electrical engineering from University of Tehran and his PhD in telecom from the University of Waterloo in Canada. He has been with the University of Toronto, Canada, Ecole Polytechnique Fédéral de Lausanne (EPFL), Switzerland, Jilin University, China, University of California, Santa Cruz, USA, and Queen's University, Kingston, Canada, where he is currently a tenured full professor and Associate Dean of Corporate Relations. Shahram has served as the Editor–in-Chief of the IEEE CJECE journal and has received more than 20 awards and scholarships including the Golden Apple teaching award as well as the Natural Sciences and Engineering Research Council of Canada’s Discovery Accelerator Supplement (NSERC DAS) award. His research interests include communications, cloud systems, big data, networks, information theory, signal processing, control, algorithms, and machine intelligence in which he holds a number of patents. He co-founded Canarmony Corp. (2014), MESH Scheduling Inc. (2018), and OPTT Inc. (2018) to apply algorithms to make life better and more harmonious. One of Shahram’s patents in the area of data storage was licensed to revolutionize a \$500b solid state storage industry. Shahram admires entrepreneurs and advises a number of ventures.
\end{IEEEbiography}

\end{document}